\documentclass[runningheads]{llncs}
\usepackage[T1]{fontenc}
\usepackage{graphicx}
\usepackage{amssymb}
\usepackage{amsmath}
\usepackage{comment}
\usepackage{xcolor}
\usepackage[noend]{algpseudocode}
\usepackage{algorithmicx,algorithm} 
\usepackage{ulem}
\usepackage{rotating}
\newcommand{\Qo}{\Delta}

\begin{document}
\title{A Novel Algorithm for Representing Positive Semi-Definite Polynomials as Sums of Squares with Rational Coefficients}
\titlerunning{Sum of Squares with Rational Coefficients}
%
\author{Zhenbing Zeng\inst{1}\orcidID{0000-0001-9728-1114} \and
Yong Huang\inst{2,3}\orcidID{0000-0001-8401-2866} \and
Lu Yang\inst{4} \and
Yongsheng Rao\inst{2,3}\orcidID{0000-0001-9615-3658}
}
\authorrunning{Z. Zeng et al.}
%
\institute{
Department of Mathematics, Shanghai University, Shanghai 200444, China. 
\email{zbzeng@shu.edu.cn}\\
\and
Guangzhou University, Institute of Computing Science and Technology, Guangzhou 510006, China\\
\and
Guangdong Province Research Center for Mathematical Educational Software, Guangzhou 510006, China\\
\email{gzjzes@qq.com, rysheng@gzhu.edu.cn}
\and 
Chengdu Institute of Computer Applications, Chinese Academy of Sciences, Chengdu $610213$, China.   
\email{luyang@casit.ac.cn}
}


%
\maketitle              
\begin{abstract}
This paper presents a novel algorithm for constructing a sum-of-squares (SOS) decomposition for positive semi-definite polynomials with rational coefficients. Unlike previous methods that typically yield SOS decompositions with floating-point coefficients, our approach ensures that all coefficients in the decomposition remain rational. This is particularly useful in formal verification and symbolic computation, where exact arithmetic is required. We introduce a stepwise reduction technique that transforms a given polynomial into a sum of ladder-like squares while preserving rationality. Experimental results demonstrate the effectiveness of our method compared to existing numerical approaches.

\end{abstract}
\section{Motivation and Background}
\label{introduction}

In 1888, Hilbert \cite{Hilbert1888} showed that every non-negative homogeneous polynomial in $n$ variables and degree 2d can be represented as sum of squares of other polynomials if and only if either (a) $n = 2$ or (b) $2d = 2$ or (c) $n = 3$ and $2d = 4$. (See also \cite{Hilbert1893,Hilbert1910,Artin1927}). 
In particularly, Any univariate polynomial $p(x)$ of degree $2d$ that satisfies $p(x)\geq 0$ for all $x\in \mathbb{R}$ can be represented as sum of squares. 
If we already know all the roots of \( p(x) = 0 \), for example,  
\[
\alpha_i + \beta_i\sqrt{-1}, \quad \alpha_i - \beta_i\sqrt{-1}, \quad i = 1,2,\ldots,d,
\]  
then we have the following SOS representation:  
\[
p(x) = \prod_{i=1}^{d} \left( (x - \alpha_i)^2 + \beta_i^2 \right).
\]
Thus, applying the transformation rule  
\begin{equation}
(A^2 + B^2) \cdot (C^2 + D^2) = (A \cdot C + B \cdot D)^2 + (A \cdot D - B \cdot C)^2,
\label{tr-1}
\end{equation} 
inductively, we can express \( p(x) \) as a sum of squares of one polynomial of degree $d$ and another one polynomial of degree $\leq d-1$.  For example, when \( 2d = 4, 6 \), we have:  
\begin{align}
&\;((x-\alpha_1)^2+\beta_1^2) \cdot ((x-\alpha_2)^2+\beta_2^2) \nonumber \\
=&\; \left( (x - \alpha_1)(x - \alpha_2) + \beta_1 \beta_2 \right)^2 
+ \uwave{\left( (x - \alpha_1) \beta_2 - \beta_1 (x - \alpha_2) \right)^2}, 
\label{eq-2}
\end{align}
and
\begin{align}
&\;\left( (x - \alpha_1)^2 + \beta_1^2 \right) \cdot \left( (x - \alpha_2)^2 + \beta_2^2 \right) \cdot \left( (x - \alpha_3)^2 + \beta_3^2 \right).\nonumber\\
=&\;\left( \left( (x - \alpha_1)(x - \alpha_2) + \beta_1 \beta_2 \right)(x - \alpha_3) + \left( (x - \alpha_1)\beta_2 - \beta_1 (x - \alpha_2) \right) \beta_3 \right)^2 \nonumber \\
&\;+ \uwave{\left( \left( (x - \alpha_1)(x - \alpha_2) + \beta_1 \beta_2 \right) \beta_3 - \left( (x - \alpha_1)\beta_2 - \beta_1 (x - \alpha_2) \right)(x - \alpha_3) \right)^2},
\label{eq-3}
\end{align}
where the express with under-wave are the square of lower degree polynomial. Obviously, they are also semi-positive-definite polynomials, thus, performing the transformation for the under-waved part, we can write them as a sum of two squares, one is of a polynomial of degree $\leq d-1$,  and other one is of a polynomial of degree $\leq d-2$. Do this recursively, we can express $f(x)$ finally as a sum of at most $d+1$ squares in the following form:
$$
p(x)=q_d(x)^2+q_{d-1}(x)^2+\ldots+q_1(x)^2+q_0^2,
$$
where $q_j(x)\,(j=d,d-1,\ldots,1)$ are zero polynomials or univariate polynomials of degree $j$, and $q_0\in \mathbb{R}$. This observation inspired the following general theorem.

\begin{theorem}
Any semi-positive definite polynomial $p(x)$
od degree $2d$ can be expressed as the sum of squares of as follows:
\begin{equation}
p(x) = {q_d(x)}^2 + {q_{d-1}(x)}^2 + \ldots + {q_1(x)}^2 + q_0^2,
\label{eq-0001}
\end{equation}
where $q_k(x)\, (1 \leq k \leq d)$ is a zero polynomial or a polynomial of degree $k$, and $q_0(x) = q_0$ is a real number.
\label{thm-1}
\end{theorem}

\begin{proof}
Without loss of generality we can assume that the leading coefficient of $p(x)$ is 1. We shall use mathematics induction to $d$ to prove the theorem. 

When $2d=2$, the polynomial is in the form of $x^2 + ax + b$, which can be factored as:
$$
x^2 + ax + b = \left(x - {a}/{2}\right)^2 + \left(b - {a^2}/{4}\right),
$$
which is a sum of squares, in view $p(x)$ wither has no real roots or has two equal real roots, which means that $\Delta(p)=a^2-4b\leq 0$. 

Assume that the theorem is proved for degree $2d$ case, i.e., any semi-positive-definite polynomial which degree is $2d$ can be expressed as a sum of squares of some degree-descending polynomials. We aim to prove that given a polynomial $p(x)$ of degree $2d+2$, $p(x)$ can also be written as the sum of squares of some degree-descending polynomials. For this, we can assume that the roots of $p(x)=0$ are as follows:  
$$
\alpha_0 + \beta_0 \sqrt{-1}, \; \alpha_0 - \beta_0 \sqrt{-1}; \quad 
\alpha_1 + \beta_1 \sqrt{-1}, \; \alpha_1 - \beta_1 \sqrt{-1}; 
$$
$$
\ldots; \quad 
\alpha_d + \beta_d \sqrt{-1}, \; \alpha_d - \beta_d \sqrt{-1},
$$
where $\alpha_i, \beta_i \; (0 \leq i \leq n)$ are real numbers. If 
$$
\beta_0=\beta_1=\cdots=\beta_n=0,
$$
then $p(x)$ is a square of a degree $d+1$ polynomial, so the conclusion is true for $p(x)$. Otherwise, we have
\begin{align}
p(x) &= \left( (x - \alpha_0)^2 + \beta_0^2 \right) \times \left( (x - \alpha_1)^2 + \beta_1^2 \right) \times \cdots \times \left( (x - \alpha_d)^2 + \beta_d^2 \right) \nonumber \\
&= (x - \alpha_0)^2(x - \alpha_1)^2 \cdots (x - \alpha_d)^2 +  r(x),
\label{eq-pxr}
\end{align}
where $r(x)$ is the sum of other products of $(x-\alpha_i)^2$ and $\beta_j^2$, namely
\begin{equation}
r(x)=\sum_{{1\leq k\leq n+1}\substack {\{I,J\}}} \beta_{i_1}^2\cdots\beta_{i_k}^2
\times 
(x-\alpha)_{j_1}^2\cdots (x-\alpha)_{j_{n+1-k}}^2,
\label{residual}
\end{equation}
where $\{I,J\}$ stands for 
$$
\{i_1,\ldots,i_k,j_1,\ldots,j_{n+1-k}\}
=\{1,2,\ldots,d+1\}.
$$
Clearly, the degree of $r(x)$ is $2(d-1)$, since its leading coefficient is 
$$
\beta_0^2+\beta_1^2+\cdots+\beta_{n}^2\not=0. 
$$
According to the inductive assumption, $r(x)$ can be represented as a sum of squares of degree ascending polynomials, which implies immediately that $p(x)$ can also be represented as a sum of squares of degree-descending polynomials as we claimed. 
This completes the proof of Theorem~\ref{thm-1}. \qed
\end{proof}

Notice that the SOS generated by Theorem~\ref{thm-1} and the SOS constructed by expressions \eqref{eq-2} and \eqref{eq-3}---although both are sum-of-squares representations of degree-descending polynomials (for shorter, DDP sum of squares or DDP-SOS in forthcoming discussion)---do not necessarily consist of the same specific polynomials. This indicates that the DDP-SOS representation of a positive semidefinite polynomial is not unique. For example, for $2d=6$, from~\eqref{thm-1} we have
$$
q_3(x)=x^3-(\alpha_1+\alpha_2+\alpha_3)x^2
+(\alpha_1\alpha_2+\alpha_1\alpha_3+\alpha_2\alpha_3)x
-\alpha_1\alpha_2\alpha_3,
$$
while \eqref{eq-3} yields
\begin{align*}
q_3(x)=&\;
\left( (x - \alpha_1)(x - \alpha_2) + \beta_1 \beta_2 \right)(x - \alpha_3) + \left( (x - \alpha_1)\beta_2 - \beta_1 (x - \alpha_2) \right) \beta_3 \\
=&\; x^3-(\alpha_1\!+\!\alpha_2\!+\!\alpha_3)x^2+(\alpha_1\alpha_2\!+\!\alpha_1\alpha_3\!+\!\alpha_2\alpha_3+\beta_1\beta_2\!-\!\beta_1\beta_2\!+\!\beta_2\beta_3)x\\
&\;-\alpha_1\alpha_2\alpha_3-\alpha_1\beta_2\beta_3+\alpha_2\beta_1\beta_3-\alpha_3\beta_1\beta_2.
\end{align*}
In general, for given seimi-positive polynomial 
$$
p(x)=x^{2d}+c_{2d-1}x^{2d-1}+c_{2d-2}x^{2d-2}+\cdots+c_1x+c_0,
$$
we have
$$
q_d(x)=x^d-\left(
\sum_{i=1}^{d}\alpha_i\right)x+
\left(\sum_{1\leq i<j\leq d} \alpha_i\alpha_j \right)\, x^{d-2}+\cdots+
(-1)^d\alpha_1\alpha_2\cdots \alpha_d.
$$
According to Vieta's theorem, we have the following equalities:
$$
\sum_{i=1}^{d}\alpha_i=
\frac{1}{2}\sum_{i=1}^{d}
\left((\alpha_1+\beta_1\sqrt{-1})+(\alpha_1-\beta_1\sqrt{-1})\right)
=-\frac{1}{2}c_{2d-1},
$$
$$
2\times \sum_{1\leq i<j\leq d} \alpha_i\alpha_j=
\left(
\sum_{i=1}^{d}\alpha_i\right)^2-
\sum_{i=1}^{d}\alpha_i^2=\frac{1}{4}c_{2d-1}^2-
(\alpha_1^2+\alpha_2^2+\cdots+\alpha_d^2)
,
$$
and
$$
\left(\alpha_1^2+\alpha_2^2+\cdots+\alpha_d^2\right)
-
\left( \beta_1^2+\cdots+\beta_{n}^2 \right)=
\frac{1}{2}\sum_{k=1}^{2d}z_k^2
=\frac{1}{2}c_{2d-1}^2-c_{2d-2},
$$
here $z_1,z_2,\ldots,z_{2d}$ are $2d$ roots of the equation $p(x)=0$. 
It leads the following identity:
\begin{equation}
 q_d(x)=x^{d}+\frac{1}{2}c_{2d-1}x^{d-1}
 +\frac{1}{2}\left(-\frac{1}{4}c_{2d-1}^2+c_{2d-2}-B_2\right)x^{d-2}+\cdots,
 \label{qdb2}
\end{equation}
\begin{equation}
q_{d-1}(x)=B_2\cdot \left(x^{d-1}+\cdots\right),
\label{qdb21}
\end{equation}
where $B_2=\beta_1^2+\beta_2^2+\ldots+\beta_d^2$. 
Apparently, if there is polynomials or rational functions $A_1,A_2,\ldots,A_d$ in terms of $c_{0},c_1,\ldots,c_{2d-1},c_{2d}$ that can express the $k$-power sums 
$$
P_k(\alpha_1,\alpha_2,\ldots,\alpha_d)=
\sum_{1\leq k\leq d}\alpha_{k}^d,
$$
then applying the inverse of Newtons-Girard formulas
(see \cite{Gathen2003})
we can express the elementary symmetric functions of  $\alpha_1,\alpha_2$, $\ldots,\alpha_d$ in terms of $c_0,c_1,\ldots,c_{2d-1},c_{2d}$,
and therefore, compute the
polynomial 
$$
p_d(x)=(x-\alpha_1)(x-\alpha_2)\cdots (x-\alpha_d)
$$
in an explicit way, and hence establish a method to express semi-positive-definite polynomial in DDP sum of squares without solving complicated equations.   
As far as we know, there were no such considerations in previous research on SOS. 
Through computer-assisted research on several specific examples, we believe that the initially determined search scope for functions such as \( A_2, \ldots, A_d \) may be too narrow. Substituting \eqref{qdb2} and~\eqref{qdb21} into \eqref{eq-0001},
we have
\begin{align*}
&\;x^{2d}+c_{2d-1}x^{2d-1}+c_{2d-2}x^{2d-2}+c_{2d-3}x^{2d-3}+\cdots\\
=&\;
(x^d+\frac{1}{2}c_{2d-1}x^{d-1}+
\uwave{\frac{1}{2}(-\frac{1}{4}c_{2d-1}^2+c_{2d-2}-B_2)}x^{d-2}+{\color{red}a_{d,d-3}}x^{d-3}+\cdots)^2
\\
&\;+ 
B_2\,(x^{d-1}+{\color{red}a_{d-1,d-2}}x^{d-2}\cdots)^2+q_{d-2}(x)^2+\cdots+q_0^2,
\end{align*}
Expanding the squares and compare the coefficients of $x^{2d}, x^{2d-1},x^{2d-2},x^{2d-3}$, we obtain the following equations:
\begin{align}
    \text{for } x^{2d-1}:\phantom{xx}&\;
    c_{2d-1}=2\times \frac{1}{2}c_{2d-1},
    \nonumber \\
    \text{for } x^{2d-2}:\phantom{xx}&\;
c_{2d-2}=(\frac{1}{2}c_{2d-1})^2+ 
2\times \uwave{\frac{1}{2}(-\frac{1}{4}c_{2d-1}^2+c_{2d-2}-B_2)}+B_2,
    \label{eq-x2d2}\\
    \text{for } x^{2d-3}:\phantom{xx}&\;
c_{2d-3}= 2\times {\color{red}a_{d,d-3}}+2\times \frac{1}{2}c_{2d-1}\times 
\uwave{\frac{1}{2}(-\frac{1}{4}c_{2d-1}^2+c_{2d-2}-B_2)}   
    \nonumber \\
    &\phantom{xxxxxx}+2B_2\times {\color{red}\alpha_{d-1,d-2}},
    \label{eq-x2d3}
\end{align}
from the equations we observed that $B_2$ can be any real number, not necessarily taking the value $\beta_1^2+\beta_2^2+\ldots+\beta_d^2$. Even more, we can choose any special value of ${\color{red}\alpha_{d,d-3}}$, then solve ${\color{red}\alpha_{d-1,d-2}}$ from the equation~\eqref{eq-x2d3}.  
Considering that if we put all 
$$
(d+1)+d+\ldots+2+1=\frac{1}{2}(d+1)(d+2)
$$
coefficients of $q_d(x),q_{d-1}(x),\ldots$, $q_1(x), q_0$ as undetermined parameters, then from \eqref{eq-0001} we will get $2d+1$ equations, which is much less than the number of variables.    
This additional finding motivates us to further explore the structural characteristics of the system of coefficient equations to find more fine properties of SOS. 

\section{SOS representation of degree-strictly-descending polynomials}
\label{ddp-sos}

In this section we prove that
any positive definite univariate polynomial always can be represented as an SOS of degree-strictly-descending polynomials. We have the following result. 

\begin{theorem}
If $p(x)$ is positive definite, and $\deg(p,x)=2d$, then  $p(x)$ can be expressed as the sum of squares of $d+1$ polynomials with decreasing degrees.  
If $p(x)$ is semi-positive definite, 
$\deg(p,x)=2d$,  
$$
p(x)=p_1(x)\times 
\prod_{j=1}^{k_0}(x-r_j)^{2}, \;
r_1,\ldots,r_{k_0}\in \mathbb{R},
$$
and $p_1(x)$ is positive definite, 
then $p(x)$ can be expressed as the sum of squares of $d-k_0+1$ polynomials with decreasing degrees from $d$ to $d-k_0$.   
\label{thm-2}
\end{theorem}

\begin{proof} First consider the case that $p(x)$ is positive definite. 
Without loss of generality, assume that the leading coefficient of $p(x)$ is 1. Since $p(x)$ is positive definite, $p(x)$ has $d$ pairs of complex conjugate roots  
\[
\alpha_i\pm\beta_i\sqrt{-1}, \quad i=1,2,\ldots,d,
\]
where $\alpha_i, \beta_i$ are real numbers, $\beta_i > 0$ for $1\leq i\leq d$.  
Thus, $p(x)$ can be rewritten as  
\begin{equation}
  p(x)=\prod_{i=1}^{d} \left[(x-\alpha_i)^2+\beta_i^2\right]=
\left(q_d(x)\vphantom{x_1^2}\right)^2+r_{d-1}(x),
\label{eq-0002}
\end{equation}
where $q_d(x)=\prod_{i=1}^{d}(x-\alpha_i)$ is a polynomial of degree $d$, and  
\begin{align}
r_{d-1}(x)&\;=\;
\sum_{k=1}^{d}\left(\frac{\beta_k}{x-\alpha_k}\,q_d(x)\right)^2\nonumber \\
&\;+\sum_{1\le k,l\le {d}, \ k\neq l}\left(\frac{\beta_k \beta_l}{(x-\alpha_k)(x-\alpha_l)}\,q_d(x)\right)^2
+\cdots+ \left(\prod_{k=0}^{d}\beta_k\right)^2.
\nonumber
\end{align}
Clearly, $\deg(r_{d-1},x)\!=\!2(d-1)$, and the leading coefficient of $r_{d-1}(x)$ is 
$b_1 = \sum_{k=1}^d \beta_k^2 >0$. 
Since $r_{d-1}(x)$ is a sum of squares of polynomials, 
and $\beta_1,\ldots,\beta_d > 0$ ,  its constant term is positive, making $r_{d-1}(x)$ strictly positive definite. 

Applying the same argument we can write 
$r(x)$ in the following form:
\begin{equation}
r_{d-1}(x)=b_1\times \left(\left(q_{d-1}(x)\right)^2+r_{d-2}(x)\right)
\label{eq-rd1}
\end{equation}
where
$$
q_{d-1}(x)= \prod_{j-1}^{d-1}(x-\alpha'_j),\quad 
r_{d-2}(x)=
\sum_{j=1}^{d-1}\left(\frac{\beta'_k}
{x-\alpha'_k}q_{d-1}(x) \right)^2+\cdots+\prod_{j=1}^{d} {\beta'_j}^2,
$$ 
and $\alpha'_j,\beta'_j>0\,(j=1,\ldots,d-1)$ are real numbers, such that $\beta'_1,\ldots,\beta'_{d-1}>0$, and  
$$
\alpha'_j+\beta'_j\sqrt{-1}, \;
\alpha'_j-\beta'_j\sqrt{-1}, \quad j=1,2,\ldots,d-1
$$
are the $2(d-1)$ complex roots of $r_{d-1}(x)=0$. Apparently, $r_{d-2}$ is a positive definite of degree $d-2$ with leading coefficient $b_2=\sum_{j=1}^{d-1}{\beta'_j}^2>0$. Therefore, we can perform this computation for $r_{d-2}(x)$ and get
\begin{equation}
r_{d-2}(x)=b_2\times \left(
(\left(q_{d-2}(x)\right)^2+r_{d-3}(x)
\right),
\label{eq-rd2}
\end{equation}
where $q_{d-2}(x)$ is a polynomial of degree $d-2$ with leading coefficient $1$, and $r_{d-3}(x)$ is a positive definite polynomial of degree $2(d-3)$ with leading coefficient $b_3>0$. 
Inductively, we have
\begin{align}
r_{d-3}(x)=&\;b_3\times \left(
(\left(q_{d-3}(x)\right)^2+r_{d-4}(x)
\right),
\label{eq-rd3}\\
\vdots\;&\nonumber \\
r_{2}(x)=&\;b_{d-2}\times \left(
(\left(q_{2}(x)\right)^2+r_{1}(x)
\right),
\label{eq-rdj}
\end{align}
where $q_{j}\,(j=d-3,\ldots,2)$ are polynomial of degree $j$ with leading coefficient $1$, 
$r_{j}\,(j=d-4,\ldots,1)$ are  positive definite polynomial of degree $2j$ with leading coefficient $b_{d-j}>0$.  In particular,
\begin{equation}
r_1(x)=b_{d-1}\times (x^2+ax+b)
b_1(x)=b_{d-1}\times \left(\left(x+\frac{a}{2}\right)^2+
(b-\frac{a^2}{4})\right)
\label{eq-rdd}
\end{equation}
is a positive definite quadratic polynomial. Let $q_1(x)=x+a/2$, $b_d=b-a^2/4$. Then from \eqref{eq-0002} to \eqref{eq-rdd} we have
\begin{align}
p(x)=q_d(x)^2
+b_1\cdot q_{d-1}(x)^2
+b_1b_2\cdot q_{d-2}(x)^2\nonumber \\
+\cdots+b_1b_2\cdots b_{d-1}q_1(x)^2+
b_1b_2\cdots b_d,
\end{align}
which shows that $p(x)$ can be expressed as a sum of square of $d+1$ polynomials of degree from $d$ to zero, as claimed in the Theorem~\ref{thm-2}.  

Now, consider the case that 
$$
p(x)=p_1(x)\times r(x), \quad 
r(x)=(x-r_1)^2\cdots (x-r_{k_0})^2,
$$
where $p_1(x)$ is a positive definite polynomial of degree $2(d-k_0)$, $r_1,\ldots,r_{k_0}$ are real numbers.
Let
\[
p_1(x)=\left(\tilde{q}_{d-k_0}(x)\right)^2+\left(\tilde{q}_{d-k_0-1}(x)\right)^2+\ldots+\left(\tilde{q}_{0}(x)\right)^2,
\]
be the SOS of $p_1(x)$ into $d-k_0+1$ degree-descending-polynomials, 
where $\tilde{q}_{k}(x)$ (for $1\leq k\leq d-k_0$) is a real-coefficient polynomial of degree $k$ with a nonzero leading coefficient,  
and $\tilde{q}_{0}(x)$ is a nonzero real number. Let
$$
q_j=\tilde{q}_{j-k_0}\cdot r(x), \; \;
j=d,d-1,\ldots,k_0.
$$
Then
$$
p(x)=q_d(x)^2+q_{d-1}(x)^2+\ldots+
q_{k_0}(x)^2.
$$
which means that $p(x)$ can be written as the sum of squares of $d-k_0+1$ polynomials with decreasing degrees,  
where the highest and lowest degrees are $d$ and $k_0$, respectively.  

{Theorem~\ref{thm-2} is proved.}  
\qed  
    
\end{proof}  

\section{An algorithm for constructing DDP-SOS representation
}
\label{algorithm}

Assume that 
$$
p(x)=x^{2d}+c_{2d-1}x^{2d-1}+\cdots+c_1x+c_0,
$$
is a positive definite polynomial of degree $2d$,  and $q_k(x)\,(k=d,d-1,\ldots,1,0)$ is a sequence of polynomials satisfying $\deg(q_k(x),x)=k$ and
\begin{equation} 
p(x)=q_d(x)^2+q_{d-1}(x)^2+\ldots+q_{1}(x)^2+q_{0}^2.
\label{eq-001}
\end{equation} 
Let 
\begin{equation}
q_k(x)=a_{k,k}x^k+a_{k,k-1}x^{k-1}+\cdots+a_{k,1}x+a_{k,0}, \quad k=d,d-1,\ldots,1,0.
\label{eq-qks}
\end{equation}
Then the DDP SOS representation~\eqref{eq-001} can be written as follows: 
$$
p(x)=\left(x^{d},x^{d-1},\cdots,x,1\right)
L\,L^T
\left(
x^{d},\,
x^{d-1},\,
\cdots,\,
x,\,
1
\right)^T,
$$
here $L$ is the following lower-triangle matrix:
\begin{equation}
L=\left(
\begin{array}{cccccc}
a_{d,d}&\\[3pt]
a_{d,d-1}&{\color{blue}a_{d-1,d-1}}&\\[5pt]
a_{d,d-2}&\;{\color{red}a_{d-1,d-2}}\;&\;{\color{blue}a_{d-2,d-2}}\;&\\[5pt]
\vdots&{\color{red}\vdots}&{\color{red}\vdots}&\begin{turn}{13} ${\color{blue}\ddots}$\end{turn}&\\[5pt]
a_{d,1}&{\color{red}a_{d-1,1}}&{\color{red}a_{d-2,1}}&{\color{red}\cdots}&{\color{blue}a_{1,1}}&\\[5pt]
a_{d,0}&a_{d-1,0}&a_{d-2,0}&\phantom{x}\cdots\phantom{x}&\phantom{x}a_{1,0}\phantom{x}&\;\;a_{0,0}
\end{array}
\right)_{\;\;(d+1)\times (d+1)}
\end{equation}
\begin{picture}(200,10)(10,0)
{\color{red}
\put(90,38){\line(1,0){7}}
\put(100,38){\line(1,0){7}}
\put(110,38){\line(1,0){7}}
\put(120,38){\line(1,0){7}}
\put(130,38){\line(1,0){7}}
\put(140,38){\line(1,0){7}}
\put(150,38){\line(1,0){7}}
\put(160,38){\line(1,0){7}}
\put(170,38){\line(1,0){7}}
\put(180,38){\line(1,0){7}}
\put(190,38){\line(1,0){7}}
\put(200,38){\line(1,0){7}}
\put(210,38){\line(1,0){7}}
\put(90,38){\line(0,1){5}}
\put(90,46){\line(0,1){5}}
\put(90,54){\line(0,1){5}}
\put(90,62){\line(0,1){5}}
\put(90,70){\line(0,1){5}}
\put(90,78){\line(0,1){5}}
\put(90,86){\line(0,1){5}}
\put(90,94){\line(0,1){5}}
\put(90,102){\line(2,-1){10}}
\put(104,95){\line(2,-1){10}}
\put(118,88){\line(2,-1){10}}
\put(132,81){\line(2,-1){10}}
\put(146,74){\line(2,-1){10}}
\put(160,67){\line(2,-1){10}}
\put(174,60){\line(2,-1){10}}
\put(188,53){\line(2,-1){10}}
\put(202,46){\line(2,-1){10}}
}
\end{picture}

\noindent Here for convenience, we 
divide the 
the entries $a_{i,j}$ into three categories:
\begin{enumerate}
    \item $B=\{a_{i,j}, i=n \text{ or } j=0\}$, contains $2d+1$ {\it border\/} entries in the first column and the last row;  
    \item $C=\{a_{i,j}, 1\leq j<i\leq d-1\}$, 
    the {\it core\/} entries (printed in red) inside the triangle area enclosed by the dashed lines;
    \item $D=\{a_{i,j}, 0<i=j<d\}$, includes the {\it diagonal\/} entries (printed in blue) in the diagonal, except $a_{d,d}$ and $a_{0,0}$. 
\end{enumerate}

Substituting $q_k(x)$ in~\eqref{eq-qks} into \eqref{eq-0001} and comparing the coefficients of $x^{2d}$, $x^{2d-1}$, $\ldots$, $x,1$ of left and right sides, 
we obtain the following equations: 

\begin{align}
  c_{2d}&=1={\color{red}\uwave{a_{d,d}^2}},
  \nonumber\\
  c_{2d-1}&=2a_{d,d}{\color{red}\uwave{a_{d,d-1}}},
  \nonumber\\
  c_{2d-2}&=2a_{d,d}{\color{red}\uwave{a_{d,d-2}}}+G_{2d-2}(C,D;\uline{a_{d,d-1}}),
  \nonumber\\
  &\phantom{x}\vdots\nonumber\\[5pt]
  c_{d+1}&=2a_{d,d}{\color{red}\uwave{a_{d,1}}}+G_{d+1}(C,D; \uline{a_{d,d-1}, \ldots, a_{d,2}}),
  \nonumber\\
  c_{d}&=2a_{d,d}{\color{red}\uwave{a_{d,0}}}+G_d(C,D; \uline{a_{d,d-1}, \ldots, a_{d,2},a_{d,1}});
  \nonumber\\
  c_{d-1}&=2a_{d-1,d-1}{\color{red}\uwave{a_{d-1,0}}}+G_{d-1}(C,D;  \uline{a_{d,d-1}, \ldots, a_{d,1}}; \uuline{a_{d,0}}),
  \nonumber\\
  c_{d-2}&=2a_{d-2,d-2}{\color{red}\uwave{a_{d-2,0}}}+G_{d-2}(C,D; \uline{a_{d,d-1}, \ldots, a_{d,1}}; \uuline{a_{d,0},a_{d-1,0}}),
  \nonumber\\ 
  &\phantom{x}\vdots\nonumber\\[5pt]
  c_1&=2a_{1,1}{\color{red}\uwave{a_{1,0}}}+G_1(C,D; \uline{a_{d,d-1}, \ldots, a_{d,1}}; \uuline{a_{d,0},\ldots, a_{2,0}}),
  \nonumber\\
  c_0&={\color{red}\uwave{a_{0,0}^2}}+a_{1,0}^2+\ldots+a_{d,0}^2,
\label{the-eqs}
\end{align}
where $G_{2d-2},\ldots,G_{d+1},G_{d},\ldots,G_{1}$ are polynomials.  
Consider the border variables, i.e.,  $a_{i,j}\in B$,  as main variables and $a_i,j\in C\cup D$ as free parameters, and 
define the following order of main variables:
$$
a_{d,d}\prec a_{d,d-1} \prec a_{d,d-2} \prec \ldots \prec  a_{d,0} \prec a_{d-1,0} \prec a_{d-2,0} \prec \ldots \prec a_{1,0} \prec a_{0,0}.
$$
Them, we can observe that equation~\eqref{the-eqs} is exactly a triangular system with respect to the main variables. In other words, the main variables that appear in the equations form the following structure:
\begin{align*}
    \;{\color{red}a_{d,d}}\\
    \;{\color{red}a_{d,d-1}},\;a_{d,d}\\
    \;{\color{red}a_{d,d-2}},\;a_{d,d-1},\;a_{d,d}\\
    \;\vdots\\
    \;{\color{red}a_{d,0}},\;a_{d,1},\;\ldots,\; a_{d,d-2},\;a_{d,d-1},\;a_{d,d}\\
    \;{\color{red}a_{d-1,0}},a_{d,0},\;a_{d,1},\;\ldots,\; a_{d,d-2},\;a_{d,d-1},\;a_{d,d}\\
    \;{\color{red}a_{d-2,0}},\;a_{d-1,0},a_{d,0},\;a_{d,1},\;\ldots,\; a_{d,d-2},\;a_{d,d-1},\;a_{d,d}\\
    \l\vdots\\
    \;{\color{red}a_{1,0}},\;\ldots,\;a_{d-2,0},\;\;a_{d-1,0},a_{d,0},\;a_{d,1},\;\ldots,\; a_{d,d-2},\;a_{d,d-1},\;a_{d,d}\\
    \;{\color{red}a_{0,0}},\;a_{1,0},\;\ldots,\;a_{d-2,0},\;\;a_{d-1,0},a_{d,0},\;a_{d,1},\;\ldots,\; a_{d,d-2},\;a_{d,d-1},\;a_{d,d}
\end{align*}
We can also see that, in fact, among the $2d+1$ equations, the first one and the last one are quadratic equations, and all other $2d-1$ can be viewed as a linear equations about their leading variables (i.e., those are printed with underwaves).  
Therefore, we can solve the equation system~\eqref{the-eqs} successively, and get the following solutions:
\begin{align*}
a_{d,d}=&\;1,\\
a_{d,i}=&\;\frac{1}{2}\left(c_{d+i}-G_{d+i}(C,D; a_{d,d-1},\ldots,a_{d,i+1})\right), \\
&\; \quad (i=d-1,\ldots,1,0);\\
a_{j,0}=&\;\frac{1}{2a_{j,j}}\left(c_j-G_j(C,D; a_{d,d-1},\ldots,a_{d,1};a_{d,0},\ldots,a_{j+1,0})\right),\\
&\; \quad (j=d-1,d-2,\ldots,1);\\
a_{0,0}^2=&\;c_0-a_{d,0}^2-a_{d-1,0}^2-\ldots-a_{1,0}^2.
\end{align*}
Notice that $a_{d,d}=1$, and $\alpha_{d,d-1}$ is a polynomial of variables in $C,D$, thus, substituting $a_{d,d},a_{d,d-1}$ into $a_{d,d-2}$ we can see that $a_{d,d-2}$ is also a polynomial of variables in $C,D$. Performing this substitution successively for $a_{d,i}$ for $i=d-3,\ldots,1,0$ we can obtain the polynomial solution of $a_{d,i}\,(i=d-1,d-2,\ldots,1,0$:
\begin{equation}
a_{d,i}=P_i(C,D), \quad i=d-1,d-2,\ldots,1,0.
\label{soladi}
\end{equation}
Let $Q_d(C,D)=P_0(C,D)$. Then, $a_{d,0}=Q_d(C,D)$.  
Substitute the above solution~\eqref{soladi} into $a_{d-1,0}$, we get a rational solution of $a_{d-1,0}$:
$$
a_{d-1,0}=Q_{d-1}(C,D),
$$
and substitute this and ~\eqref{soladi} into $a_{d-2,0}$, we get
the rational solution of $a_{d-2,0}$
$$
a_{d-2,0}=Q_{d-2}(C,D),
$$
perform this substitution successive, finally we 
$$
a_{1,0}=Q_{1}(C,D).
$$
Finally, we get the rational solution of $a_{0,0}$:
\begin{align*}
a_{0,0}^2=&\;c_0-a_{d,0}^2-a_{d-1,0}^2-\ldots-a_{1,0}^2\\
&=\;c_0-\sum_{j=1}^{d}Q_j(C,D)^2,\\
&=:{\Qo}(C,D).
\end{align*}
where
$$
C=\{a_{i,j},\;| 1\leq j<j\leq d\}, 
\quad
D=\{d_{1,1},d_{2,2},\ldots,a_{d-1,d-1}\}, 
$$
as defined in previous.  

The above analysis shows that if a semi-positive polynomial $p(x)$ of degree $2d$ can be written as a DDP-SOS in form
\begin{align}
p(x)=&\; (a_{0,0}x^d+a_{d,d-1}x^{d-1}\cdots+a_{d,0})^2\nonumber \\
&\;+(a_{1,1}x^{d-1}+\cdots+a_{d-1,0})^2+\cdots+
(a_{1,1}x+a_{1,0})^2+a_{0,0}^2,
\label{ddp-sos-a}
\end{align}
such that all $a_{i,j}$ in $D$ are positive, then we can construct 
an instance of this special kind DDP-SOS by at first solving 
$a_{i,j}$ in $C\cup D$ from the following system of inequalities
\begin{equation}
\left\{
\begin{array}{rl}
    &\;a_{1,1}> 0,\; \cdots, a_{a-1,d-1}> 0, \\[3pt]
    &\;{\Qo}(a_{1,1}, \ldots, a_{d-1,d-1}; a_{i,j}|_{1\leq j<i\leq d})>0,
\end{array}\right.
\label{the-ineqs}
\end{equation}
and then solving  $a_{i,j}$ in the set $L$ from the triangular equation system \eqref{the-eqs}.  
\medskip

To conclude this section we use the above method to seek the SOS representation of degree-strictly-descending polynomials of the following polynomial:
$$
f(x)=x^6-x^5-2x^4+x^3+x^2+1.
$$
Applying the Sturm theorem it can be easily check that the given polynomial is positive definite. 
Using software YAMILP we can find an approximate SOS representation 
$$
f=f_1^2+f_2^2+f_3^2+f_4^2
$$
where
\begin{align*}
&f_1 =-0.1398202913 + 1.4092x + 0.4843x^2 - 0.9765x^3, \\
&f_2 =-0.9584996211 + 0.2613x + 0.5896x^2 - 0.0526x^3,\\
&f_3 =-0.2470324237 + 0.2202x - 0.3213x^2 + 0.1938x^3,\\
&f_4 = 0.02652817796 + 0.0386x + 0.0532x^2 + 0.0782x^3.
\end{align*}
Our target is to find $4$ polynomials which degrees are $3,2,1,0$, respectively
so that $f(x)$ is the sum of squares of these polynomials.
We don't know whether or not the given $f(x)$ has a such DDP-SOS representation. 

In our method, we first assume that
\begin{align*}
&x^6-x^5-2x^4+x^3+x^2+1\\
=\,&
\left(x^3,x^2,x,1\right)
\left(\begin{array}{cccc}
a_{3,3}&&&\\
a_{3,2}&a_{2,2}&&\\
a_{3,1}&a_{2,1}&a_{1,1}&\\
a_{3,0}&a_{2,0}&a_{1,0}&a_{0,0}
\end{array}\right)
\left(\begin{array}{cccc}
a_{3,3}&a_{3,2}&a_{3,1}&a_{3,0}\\
&a_{2,2}&a_{2,1}&q_{2,0}\\
&&a_{1,1}&a_{1,0}\\
&&&a_{0,0}
\end{array}\right)
\left(\begin{array}{c}
x^3\\x^2\\x\\1
\end{array}\right)\\
=\,&\left(a_{3,3}x^3+a_{3,2}x^2+a_{3,1}x+a_{3,0}\right)^2\\
&\;
+\left(a_{2,2}x^2+a_{2,1}x+a_{2,0}\right)^2
+\left(a_{1,1}x+a_{1,0}\right)^2
+\left(a_{0,0}\right)^2.
\end{align*}
In this problem the undeterminate $a_{i,j}$ can be decomposed into $B,C,D$ as follows:
$$
B=\,\left(a_{3,3},a_{3,2},a_{3,1},a_{3,0}; a_{2,0},a_{1,0},a_{0,0}\right),\quad
C=\,\left(a_{2,1}\right), \quad 
D=\,\left(a_{2,2},a_{1,1}\right).
$$
Then we can construct the following triangular equations:
\begin{align}
{{\color{red} a_{3,3}}}^{2}-1=0,\nonumber \\
2\,{\color{red} a_{3,2}}\,{ a_{3,3}}+1=0,\nonumber \\
2\,{\color{red} a_{3,1}}\,{ a_{3,3}}+{{ a_{3,2}}}^{2}+{{ a_{2,2}}}^{2}+2=0,\nonumber \\
2\,{\color{red} a_{3,0}}\,{ a_{3,3}}+2\,{ a_{3,1}}\,{ a_{3,2}}+2\,{ a_{2,1}}\,{ a_{2,2}}-1=0,\nonumber \\
2\,{\color{red} a_{2,0}}\,{ a_{2,2}}+{{ a_{1,1}}}^{2}+{{ a_{2,1}}}^{2}+2\,{ a_{3,0}}\,{ a_{3,2}}+{{ a_{3,1}}}^{2}-1=0,\nonumber \\
2\,{\color{red} a_{1,0}}\,{ a_{1,1}}+2\,{ a_{2,0}}\,{ a_{2,1}}+2\,{ a_{3,0}}\,{ a_{3,1}}\phantom{space here for purpose}=0,\nonumber \\
{{\color{red} a_{0,0}}}^{2}+{{ a_{1,0}}}^{2}+{{ a_{2,0}}}^{2}+{{ a_{3,0}}}^{2}-1\phantom{please don't delete the space here}=0.\nonumber \\
\label{the-eqs7}
\end{align}
The variables printed in red belong to $B$. Taking them as the main variables. The first $6$ equations can be viewed as linear equations with respect to the main variables, therefore, we have the following solutions: 
\begin{align}
a_{3,3}\,&=1,\nonumber \\
a_{3,2}\,&=-\frac{1}{2a_{3,3}}=-\frac{1}{2},\nonumber \\
a_{3,1}\,&=-\frac{2+a_{3,2}^2+a_{2,2}^2}{2a_{3,3}}=-\frac{9+4a_{2,2}^2}{8},\nonumber \\
a_{3,0}\,&=-\frac{2a_{3,1}a_{3,2}+2a_{2,1}a_{2,2}-1}{2a_{3,3}}=-\frac{1+4a_{2,2}^2+16a_{2,1}a_{2,2}}{16},\nonumber \\
a_{2,0}\,&=-\frac{a_{1,1}^2+a_{2,1}^2+2a_{3,0}a_{3,2}+a_{3,1}^2-1}{2a_{2,2}}\nonumber \\
\,&=-\frac{16\,{{ a_{2,2}}}^{4}+64\,{{ a_{1,1}}}^{2}+64\,{{ a_{2,1}}}^{2}+64\,{ a_{2,1}}\,{ a_{2,2}}+88\,{{ a_{2,2}}}^{2}+21}
{128a_{2,2}},\nonumber \\
a_{1,0}\,&=-\frac{a_{2,0}a_{2,1}+a_{3,0}a_{3,1}}{a_{1,1}}=:A_{10}(a_{2,2},a_{1,1};a_{2,1}),
\label{soladi76} 
\end{align}
where $A_{10}$ is the following rational function:
\begin{align*}
{ A_{10}}=&\,\left(-48\,{ a_{2,1}}\,{{ a_{2,2}}}^{4}-16\,{{ a_{2,2}}}^{5}
 +64\,{{ a_{1,1}}}^{2}{ a_{2,1}}+64\,{{ a_{2,1}}}^{3}+64\,{{ a_{2,1}}}^{2}{ a_{2,2}}\right.\\
&\,\left.\phantom{x}-56\,{ a_{2,1}}\,{{ a_{2,2}}}^{2}-40\,{{ a_{2,2}}}^{3}+21\,{ a_{2,1}}-9\,{ a_{2,2}}\right)
/\left({128\,{ a_{2,2}}\,{ a_{1,1}}}\right). 
\end{align*}
Substitute the solution in \eqref{soladi76} into the last equation of the system \eqref{the-eqs7}, we get
\begin{align}
(a_{0,0})^2\,&=1-a_{1,0}^2-a_{2,0}^2-a_{3,0}^2
=\frac{P_{42}(a_{2,2},a_{1,1},a_{2,1})}{16384a_{1,1}^2a_{2,2}^2},
\label{soladi70}
\end{align}
here the polynomial $P_{42}$
is the following polynomial with $42$ monomials:
\begin{align*}
P_{42}=&-256\,{{ a_{1,1}}}^{2}{{ a_{2,2}}}^{8}-2304\,{{ a_{2,1}}}^{2}{{ a_{2,2}}}^
{8}-1536\,{ a_{2,1}}\,{{ a_{2,2}}}^{9}-256\,{{ a_{2,2}}}^{10}\\
&-2048\,{{
 a_{1,1}}}^{4}{{ a_{2,2}}}^{4}-12288\,{{ a_{1,1}}}^{2}{{ a_{2,1}}}^{2}{{
 a_{2,2}}}^{4}-8192\,{{ a_{1,1}}}^{2}{ a_{2,1}}\,{{ a_{2,2}}}^{5}\\
&-3840\,{{
 a_{1,1}}}^{2}{{ a_{2,2}}}^{6}+6144\,{{ a_{2,1}}}^{4}{{ a_{2,2}}}^{4}+8192\,{{ a_{2,1}}}^{3}{{ a_{2,2}}}^{5}-3328\,{{ a_{2,1}}}^{2}{{ a_{2,2}}}^{6}\\
&-5632\,{ a_{2,1}}\,{{ a_{2,2}}}^{7}-1280\,{{ a_{2,2}}}^{8}-4096\,{{ a_{1,1}
}}^{6}-12288\,{{ a_{1,1}}}^{4}{{ a_{2,1}}}^{2}\\
&-8192\,{{ a_{1,1}}}^{4}{
 a_{2,1}}\,{ a_{2,2}}-11264\,{{ a_{1,1}}}^{4}{{ a_{2,2}}}^{2}-12288\,{{
 a_{1,1}}}^{2}{{ a_{2,1}}}^{4}\\
&-16384\,{{ a_{1,1}}}^{2}{{ a_{2,1}}}^{3}{
 a_{2,2}}-8192\,{{ a_{1,1}}}^{2}{{ a_{2,1}}}^{2}{{ a_{2,2}}}^{2}-8192\,{{
 a_{1,1}}}^{2}{ a_{2,1}}\,{{ a_{2,2}}}^{3}\\
&-8928\,{{ a_{1,1}}}^{2}{{ a_{2,2}
}}^{4}-4096\,{{ a_{2,1}}}^{6}-8192\,{{ a_{2,1}}}^{5}{ a_{2,2}}+3072\,{{
 a_{2,1}}}^{4}{{ a_{2,2}}}^{2}\\
&+12288\,{{ a_{2,1}}}^{3}{{ a_{2,2}}}^{3}+
4000\,{{ a_{2,1}}}^{2}{{ a_{2,2}}}^{4}-4672\,{ a_{2,1}}\,{{ a_{2,2}}}^{5}-
1888\,{{ a_{2,2}}}^{6}\\
&-2688\,{{ a_{1,1}}}^{4}-5376\,{{ a_{1,1}}}^{2}{{
 a_{2,1}}}^{2}-1536\,{{ a_{1,1}}}^{2}{ a_{2,1}}\,{ a_{2,2}}+12624\,{{ a_{2,2}}}^{2}{{ a_{1,1}}}^{2}\\
&-2688\,{{ a_{2,1}}}^{4}-1536\,{{ a_{2,1}}}^{3}{
 a_{2,2}}+3504\,{{ a_{2,1}}}^{2}{{ a_{2,2}}}^{2}+672\,{ a_{2,1}}\,{{ a_{2,2}}}^{3}\\
&-720\,{{ a_{2,2}}}^{4}-441\,{{ a_{1,1}}}^{2}-441\,{{ a_{2,1}}}^{
2}+378\,{ a_{2,1}}\,{ a_{2,2}}-81\,{{ a_{2,2}}}^{2}
\end{align*}
In the final step, we need to solve the 
inequality system formed by 
\begin{equation}
a_{2,2}>0,\quad a_{1,1}>0, \quad P_{42}(a_{2,2},a_{1,1},a_{2,1})>0,
\label{semi-algebraic-system-42}
\end{equation}
In general, this kind of inequalities can be solved by the cylindrical algebraic decomposition (CAD) method or PCAD (Partial CAD) method, the algorithms have been implemented in various numerical computation software and computer algebra software, see for example \cite{Gathen2003,Grigoriev1988,Khachiyan1997,lihongzhi2010,xia2016}. We will not explain the details here. 
Here we just use a practical method to solve the specific system~\eqref{semi-algebraic-system-42}. Namely, we set the core variables in the set $C$ to be zero, i.e., $a_{2,1}=0$,  and try to find solution of the following simpler inequality
\begin{align*}
P_{14}=&
-441\,{a_{1,1}}^{2}-81\,{a_{2,2}}^{2}
+12624\,{a_{2,2}}^{2}{a_{1,1}}^{2}-720\,{a_{2,2}}^{4}-2688\,{a_{1,1}}^{4}\\
&-4096\,{a_{1,1}}^{6}-11264\,{a_{1,1}}^{4}{a_{2,2}}^{2}-8928\,{a_{1,1}}^{2}{a_{2,2}}^{4}-1888\,{a_{2,2}}^{6}-1280\,{a_{2,2}}^{8}\\
&-2048\,{a_{1,1}}^{4}{a_{2,2}}^{4}-3840\,{a_{1,1}}^{2}{a_{2,2}}^{6}
-256\,{a_{1,1}}^{2}{a_{2,2}}^{8}-256\,{a_{2,2}}^{10}
>0.
\end{align*}
Notice that $P_{14}$  is a degree-10 bi-variate polynomial with 14 monomials, and only one among them has positive coefficient. So we set  $a_{1,1}=a_{2,2}=t$, where $t$ is a positive real number. Then inequality
$P_{14}>0$ becomes the following one:
$$
g(t):=-512t^{10}-7168t^8-26176t^6+9216t^4+522t^2>0.
$$
Now it is easy to see $g(1/2)=8>0$. 
Thus, we have found a solution of 
\eqref{semi-algebraic-system-42} as follows:
$$
C^*=\left(a_{2,1}=0\right), \quad 
D^*=\left(a_{2,2}^*={1}/{2}, a_{1,1}^*={1}/{2}\right).
$$
Substitute this to \eqref{soladi76} and \eqref{soladi70}, we obtain
\begin{align*}
B^*=&\; \left(a_{3,3}=1, \, a_{3,2}=-{1}/{2}, \, a_{3,1}=-{5}/{4}, \, a_{3,0}=-{1}/{8}; \; \right.\\
&\phantom{xxx}\left. a_{2,0}=-{15}/{16}, \, a_{1,0}=-{5}/{16}, \, a_{0,0}={1}/{128}
\right).
\end{align*}
This yields the following DDP-SOS of given $f(x)$: 
$$
p(x)=\left(x^3-\frac{1}{2}x^2-\frac{5}{4}x-\frac{1}{8}\right)^2+
\left(\frac{1}{2}x^2-\frac{15}{16}\right)^2+
\left(\frac{1}{2}x-\frac{5}{16}\right)^2+
\frac{1}{128}.
$$

The algorithm can be used to reconstruct the DDP SOS of semi-definite polynomial. For example, following this method, we obtain the following DDP-SOS: 
\[
x^6+x^2 = \left( x^3- \frac{1}{2}x \right)^2+x^4+\frac{3}{4}\,x^2.
\]

The Algorithm can also used to polynomials which coefficients contain parameters. For example, for the polynomial 
$p(x)=7x^6-a\,x+3$, we can get
$$
7x^6\!-\!ax\!+\!3=7\left[\left( {x}^{3}\!-\!\frac{x}{2} \right) ^{2}\!+\! \left( {x}^{2}\!-\!{\frac{1}{2}}
 \right) ^{2}\!+\! 3\left( \frac{x}{2}\!-\!\frac{a}{21} \right) ^{2}\right]\!+\!\frac{105-4a^2}{84}.
$$
which implies that $105-4a^2>0$, the polynomial is positive definite.

\section{SOS of Rational Coefficients and the Cholesky Factorization}
\label{rational-SOS}

In this section, we first prove a theorem concerning the SOS (Sum of Squares) of rational coefficients, and then discuss the potential applications of the Cholesky factorization for finding such representations. The theorem is stated as follows.

\begin{theorem}
Every positive definite polynomial \( p(x) \) of degree \( 2d \) with rational coefficients can be expressed as the sum of one positive rational number and \( d \) squares of polynomials with rational coefficients, where the degrees of these polynomials strictly decrease from \( d \) to 1.
\label{thm-3}
\end{theorem}

\begin{proof}
    According to the algorithm presented in the previous section, given any solution \((C^{\ast}, D^{\ast})\) of the diagonal and core variables in the inequality system \eqref{the-ineqs}, we can compute the border variables \( B^{\ast} \), thereby constructing an SOS of degree-strictly-descending polynomials. By Theorem~\ref{thm-2}, the set of \((C, D)\) solutions is non-empty. Therefore, the solution set of the inequality system \eqref{the-ineqs} forms a non-empty open set in the space \( \mathbb{R}^{d-1}_{>0} \times \mathbb{R}^{(d-1)(d-2)/2} \). By the density of rational numbers, we can always find a rational solution \( (D^{\ast}, V^{\ast}) \) within this set. Since in the algorithm the variables \( a_{0,0}^2 \) and \( a_{i,j} \) in \( B^{\ast} \setminus \{a_{0,0}\} \) are determined by rational functions of \( C^{\ast}, D^{\ast} \), and the coefficients of \( p(x) \), we ultimately obtain a rational lower triangular matrix \( L \), which leads to an SOS representation 
    in form~\eqref{ddp-sos-a}
where \( a_{0,0}^2 \in \mathbb{Q} \) and \( a_{i,j} \in \mathbb{Q} \) for all \( (i,j) \neq (0,0) \).
Theorem~\ref{thm-3} is therefore proved. 
\qed 
\end{proof}

It is worth noting that, given any polynomial \( p(x) \) of degree \( 2d \):
$$
p(x) = x^{2d} + c_{2d-1} x^{2d-1} + \dots + c_1 x + c_0,
$$
the symmetric matrix \( B \) of order \( (d+1) \):
$$
B = \begin{pmatrix}
b_{d,d} & b_{d-1,d} & \cdots & b_{1,d} & d_{0,d} \\
b_{d,d-1} & b_{d-1,d-1} & \cdots & b_{1,d-1} & d_{0,d-1} \\
\vdots & \vdots & \ddots & \vdots & \vdots \\
b_{d,1} & b_{d-1,1} & \cdots & d_{1,1} & b_{0,1} \\
b_{d,0} & b_{d-1,0} & \cdots & d_{1,0} & b_{0,0}
\end{pmatrix}
$$
satisfies the equality
$$
p(x) = (x^d, x^{d-1}, \ldots, x, 1) B (x^d, x^{d-1}, \ldots, x, 1)^T = X B X^T,
$$
and forms an algebraic surface in \( \mathbb{R}^{d(d+1)/2} \) of co-dimension \( 2d+1 \). 
Clearly, the positive definiteness of \( p(x) \) does not necessarily imply that \( XBY^T > 0 \) for all
$$
X = (x^d, x^{d-1}, \ldots, x, 1), \quad Y = (y^d, y^{d-1}, \ldots, y, 1).
$$
However, if a symmetric matrix \( B \) satisfying \( p(x) = XBX^T \) happens to be positive semidefinite, then the Cholesky factorization algorithm (see \cite{Gathen2003,Horn1985}) can be used to construct a unique lower-triangular matrix \( L \) with positive diagonal entries. Therefore, we can obtain an SOS of \( p(x) \) with degree-strictly-descending polynomials.

\section{Examples of DDP SOS of rational coefficients}

{\bf Example 1}.
The following SOS representation was given in \cite{Menini2015}.
\begin{flalign}\label{eq-00007-0}
&
~~~~{x}^{4}+2\,{x}^{3}-18\,{x}^{2}-12\,x+117\notag\\
&~~~~~~~~= \left( {x}^{2}+x-\frac{217046}{21315} \right)
^{2}+ \frac{29107}{21315}\left( x+\frac{89156}{29107} \right) ^{2}+\frac{6597622612523}{13224160752075}.
&
\end{flalign}
We apply our method to reconstruct SOS of degree descending polynomials. Assume that
$$
{x}^{4}+2\,{x}^{3}-18\,{x}^{2}-12\,x+117
=\left(x^2+a_{2,1}x+a_{2,0}\right)^2+\left(a_{1,1}x+a_{1,0}\right)^2+\left(a_{0,0}\right)^2,
$$
Then we have:
\begin{align*}
-2+2\,{a_{2,1}}=0,\\
{{a_{1,1}}}^{2}+{{a_{2,1}}}^{2}+2\,{a_{2,0}}+18=0,\\
2\,{a_{1,0}}\,{a_{1,1}}+2\,{a_{2,0}}\,{a_{2,1}}+12=0,\\
{{a_{0,0}}}^{2}+{{a_{1,0}}}^{2}+{{a_{2,0}}}^{2}-117=0.
\end{align*}
Solve the above equations, we obtain the following solution:
$$
a_{2,1}=1,\quad   {a_{2,0}}=-\frac{1}{2}(a_{1,1}^2+19), \quad a_{1,0}=\frac{a_{1,1}^2+7}{2a_{1,1}}.
$$
substitute to ${\Qo}(C,D)$ we get the following inequality system:
$$
a_{0,0}^2=-\frac{a_{1,1}^6+39a_{1,1}^4-93a_{1,1}^2+49}{4a_{1,1}^2}>0, \quad a_{1,1}>0.
$$
Obviously, we can take $a_{1,1}^*=1$. ´and therefore,
$$
a_{2,2}^*=1, \;
a_{2,1}^*=1, \;
a_{2,0}^*=-10;\quad
a_{1,0}^*=4;\quad
a_{0,0}^*=1.
$$
This leads to the following SOS representation:
\begin{align}\label{eq-00008-0}
&
~~~~{x}^{4}+2\,{x}^{3}-18\,{x}^{2}-12\,x+117= \left( {x}^{2}+x-10 \right)
^{2}+ \left( x+4 \right) ^{2}+1.
&
\end{align}
Using YAMILP we can obtain the following result:
$$
{x}^{4}+2\,{x}^{3}-18\,{x}^{2}-12\,x+117
=f_1^2+f_2^2+f_3^2,
$$
where
\begin{align*}
&f_1 = -10.81617267 + 0.5677  x + 0.9384  x^2,\\
&f_2 = -0.1019059252 - 1.3809 x - 0.3392  x^2,\\
&f_3 = 0.004888292138 - 0.0167 x + 0.0664  x^2.
\end{align*}

{\bf Example 2. } 
{This problem is from \cite{Feng2020}:
$$
{x}^{6}-2\,{x}^{5}+4\,{x}^{4}-6\,{x}^{3}+6\,{x}^{2}-4\,x+2
$$

Using our method, we set
$$
D=\left(a_{2,2},a_{1,1}\right), \quad C=\left(a_{2,1}\right).
$$
Then ${\Qo}(C,D)$ is a rational function of $a_{2,2},a_{1,1},a_{2,1}$, for saving place we omit here. We search the solution of the inequalities 
$$
S=\{a_{2,2}>0, a_{1,1}>0, {\Qo}(a_{2,2},a_{1,1},a_{2,1})>0\}
$$
from the following grids:
$$
\left\{1,2,3,4\right\}\times \left\{1,2,3,4\right\}
\times \left\{0,1,-1,2,-2, 3,-3, 4,-4\right\}.
$$
We obtain the following solution:
$$
a_{2,2}^*=1, \quad a_{1,1}^*=1; \quad a_{2,1}^*=-1,
$$
and therefore,
$$
a_{3,2}^*=-1, \;
a_{3,1}^*=1, \;
a_{3,0}^*=-1, \;
a_{2,0}^*=\frac{1}{2}, \;
a_{1,0}^*=-\frac{1}{2}, \;
a_{0,0}^*=\frac{\sqrt{2}}{2}.
$$
This immediately leads to the following SOS representation: 
\begin{align}\label{eq-00008-0-20}
&
~~~~{x}^{6}-2\,{x}^{5}+4\,{x}^{4}-6\,{x}^{3}+6\,{x}^{2}-4\,x+2\notag\\
&~~~~~~~~= \left( {x}
^{3}-{x}^{2}+x-1 \right) ^{2}+ \left( {x}^{2}-x+\frac{1}{2} \right) ^{2}+
 \left( x-\frac{1}{2} \right) ^{2}+\frac{1}{2}.
&
\end{align}
Using YAMILP we can represent the given polynomial as $f_1^2+f_2^2+f_3^2+f_4^2$, where
\begin{align*}
&f_1 = -0.8432833022 + 2.1732  x - 1.5698  x^2+ 0.2905  x^3,\\
&f_2 = -0.8665507045 + 0.5456  x + 1.1222  x^2- 0.5327  x^3,\\
&f_3 = -0.7043361998 - 0.3787  x - 0.0093  x^2+ 0.7382  x^3,\\
&f_4 = 0.2046305611 + 0.1889  x + 0.2061  x^2+ 0.2947  x^3.
\end{align*}

{\bf Example 3. } Find SOS representation of the following polynomial:
$$
{x}^{6}-6\,{x}^{5}+14\,{x}^{4}-18\,{x}^{3}+17\,{x}^{2}-12\,x+4.
$$

This is a semi-positive definite polynomial. In fact, we have
 $${x}^{6}-6\,{x}^{5}+14\,{x}^{4}-18\,{x}^{3}+17\,{x}^{2}-12\,x+4=\left ({x}^{2}+1\right )\left (x-1\right )^{2}\left (x-2\right )^{2}.$$
Applying our method, we can obtain the following solution:
\begin{align}\label{008}
&
~~~~{x}^{6}-6{x}^{5}+14{x}^{4}-18{x}^{3}+17{x}^{2}-12x+4\nonumber\\
&=
\left ({x}^{3}-3{x}^{2}+2x\right )^{2}+\left ({x}^{2}-3x+2
\right )^{2}.
&
\end{align}

YAMILP expresses the polynomials into the sum of squares of the following four polynomials:
\begin{align*}
&f_1 = -1.414253975\! +\! 3.5355  x\! -\! 2.8283  x^2\!+\! 0.7072  x^3,\\
&f_2 = -1.414173135\! +\! 0.7070  x\! +\! 1.4142  x^2\!-\! 0.7070  x^3,\\
&f_3 = -0.000247577789\! -\! 1.8127\!\times\! 10^{-4}  x \!-\! 4.8614\!\times\! 10^{-5}  x^2\!+\! 2.1669\!\times\! 10^{-4}  x^3,\\
&f_4 = 3.717033734\times 10^{-5} \!+\! 4.5315\!\times\! 10^{-5}  x\! +\! 6.1611\!\times\! 10^{- 5}  x^2\!+\! 9.4200\times 10^{-5}  x^3.
\end{align*}

{\bf Example 4. } Find the DDP SOS of  
$$
(x^2+1)(x-1)^2(x-2)^2+1={x}^{6}-6\,{x}^{5}+14\,{x}^{4}-18\,{x}^{3}+17\,{x}^{2}-12\,x+5.
$$

Applying our method, we obtain the following solutions:
\begin{align}\label{eq-special-1}
&
~~~~(x^2+1)(x-1)^2(x-2)^2+1,\notag\\
&~~~~~~~~=x^6-6x^5+14x^4-18x^3+17x^2-12x+5,\notag\\
&~~~~~~~~=
\left ({x}^{3}-3\,{x}^{2}+2\,x\right )^{2}+\left ({x}^{2}-3\,x+2
\right )^{2}+1,\\
&~~~~~~~~=\left( {x}^{3}-x \right) ^{2}+ \left( {x}^{2}-{\frac{9}{8}} \right) ^
{2}+\frac{1}{4}\,{x}^{2}+{\frac{47}{64}},\\
&~~~~~~~~ = \left(x^3-3x^2+2x\right)^2+\left(x^2{-3}x+\frac{3}{2}\right)^2+\left(x-\frac{3}{2}\right)^2+\frac{1}{2}, &
\\
&~~~~~~~~=\left (
{x}^{3}-3\,{x}^{2}+{\frac {19}{8}}\,x-{\frac {7}{8}}\right )^{2}+
\left (\frac{1}{2}\,{x}^{2}-2\,x+{\frac {119}{64}}\right )^{2}\nonumber \\
&~~~~~~~~\phantom{=x}+\left (\frac{1}{2}\,x-{
\frac {13}{32}}\right )^{2}+{\frac {2507}{4096}}.
\end{align}

{\bf Example 5. }
Find DDP-SOS of ${x}^{6}+1$.
In this example, we have obtained 
$$
L=\left(\begin{array}{cccc}
1&&&\\
0&1/2&&\\
-1/8&0&1/2&\\
0&-17/64&0&a_{0,0}
\end{array}\right),
\quad
a_{0,0}^2=\frac{3807}{4096},
$$
which leads to:
\begin{align}\label{eq-00008-0-01}
&
~~~~{x}^{6}+1= \left( {x}^{3}-\frac{1}{8}x \right) ^{2}+ \left( \frac{1}{2}\,{x}^{2}-{\frac
{17}{64}} \right) ^{2}+\frac{1}{4}\,{x}^{2}+{\frac{3807}{4096}}.
&
\end{align}
Using YAMILP we can get $x^6+1=f_1^2+f_2^2+f_3^2+f_4^2$, where:
\begin{align*}
&f_1 = -0.7486  x + 0.8923  x^3,\\
&f_2 = -0.892288357 + 0.7486  x^2,\\
&f_3 = 0.4514659322 + 0.5382  x^2,\\
&f_4 = 0.5382  x + 0.4515  x^3.
\end{align*}


{\bf Example 6}: A sparse polynomial ${x}^{10}-x+1$. 
We search following lower-triangular matrix such that $p(x)=XLL^TX^T$ for $X=(x^5,\ldots,x,1)$:  
$$
L=\left(\begin{array}{cccccc}
{\color{red}a_{5,5}}&&&&&\\
{\color{red}a_{5,4}}&{\color{blue}a_{4,4}}&&&&\\
{\color{red}a_{5,3}}&a_{4,3}&{\color{blue}a_{3,3}}&&&\\
{\color{red}a_{5,2}}&a_{4,2}&a_{3,2}&{\color{blue}a_{2,2}}&&\\
{\color{red}a_{5,1}}&a_{4,1}&a_{3,1}&a_{2,1}&{\color{blue}a_{1,1}}&\\
{\color{red}a_{5,0}}&{\color{red}a_{4,0}}&{\color{red}a_{3,0}}&{\color{red}a_{2,0}}&{\color{red}a_{1,0}}&{\color{red}a_{0,0}}
\end{array}\right).
$$
We divide the parameters into following three groups:
\begin{align*}
&B=\left(
a_{5,5},\,a_{5,4},\,a_{5,3},\,a_{5,2},\,a_{5,1},\,a_{5,0},\,
a_{4,0},\,a_{3,0},\,a_{2,0},\,a_{1,0},\,a_{0,0}\right), \\
&D=\left(a_{4,4},\,a_{3,3},\,a_{2,2},\,a_{1,1}\right),
\quad
C=\left(a_{4,3},\,a_{4,2},\,a_{4,1},\,a_{3,2},\,a_{3,1},\,a_{2,1}\right).
\end{align*}

The inequality ${\Qo}(C,D)>0$ involves $10$ variables, therefore, it's difficult to find solutions with the cylindrical algebraic decomposition. 
We set variables in $C$ to zero:
$$
C^{\ast}:=\left(a_{4,3}^*=a_{4,2}^*=a_{4,1}^*=a_{3,2}^*=a_{3,1}^*=a_{2,1}^*=0\right),
$$
and seek solutions for variables in 
$D=\left(a_{4,4},a_{3,3},a_{2,2},a_{1,1}\right)$  
in the following grid
$$
\left\{\frac{1}{2},1,\frac{3}{2},2,\frac{5}{2},3,\frac{7}{2},4\right\}^4.
$$
Finally, we found the following solution of 
${\Qo}(C^{\ast},D)>0, D>0$:
$$
D^{\ast}=\left(a_{4,4}^*=1,\; a_{3,3}^*=\frac{1}{2},\; a_{2,2}^*=1,\;a_{1,1}^*=1\right),
$$
Substitute the solutions to $p(x)=XLL^TX^T$, we solve border variables $a_{i,j}\in B$ and obtain the following sum of squares representation: 
\begin{align}\label{eq-0000002}
&
~~~~{x}^{10}-x+1\notag\\
&~~~~~~~~= \left( {x}^{5}\!-\!\frac{1}{2}{x}^{3}\!-\!\frac{1}{4}x \right) ^{2}\!+\! \left( {x}
^{4}\!-\!\frac{5}{8} \right) ^{2}\!+\!\frac{1}{4}{x}^{6}\!\notag\\
&~~~~~~~~=+\! \left( {x}^{2}\!-\!{\frac{17}{32}}
 \right) ^{2}\!+\! \left( x\!-\!\frac{1}{2} \right) ^{2}\!+\!{\frac{79}{1024}}.
&
\end{align}
Without pre-assume $C=C^{\ast}$, we cam also search solution of ${\Qo}(C,D)>0$ directly from the following grid:
$$
\left\{\frac{1}{2},1,\frac{3}{2},2,\frac{5}{2},3,\frac{7}{2},4\right\}^4\times
\left\{0,\frac{1}{2},-\frac{1}{2},1,-1,\frac{3}{2},-\frac{3}{2},2,-2\right\}^6,
$$
by testing ${\Qo}(C,D)>0$ for the $8^4\times 9^6=2,176,782,336$ points in this grid. Following is a sample of solutions: 
$$
D^*=\left(1,\,1,\,\frac{1}{2},\,\frac{1}{2}\right), \quad
V^*=\left(-1,\,0,\,\frac{1}{2},\,\frac{1}{2},\,0,\,\frac{1}{2}\right).
$$
which leads to the following SOS representation:
\begin{align}\label{eq-0000001}
&
~~~~{x}^{10}-x+1\notag\\
&~~~~~~~~= \left( {x}^{5}-\frac{1}{2}\,{x}^{3}+{x}^{2}-\frac{3}{4}\,x-\frac{1}{4} \right) ^
{2}+ \left( {x}^{4}-{x}^{3}+\frac{1}{2}x-\frac{5}{8} \right) ^{2}\notag\\
&~~~~~~~~~~~~+ \left( \frac{1}{2}\,{x}^{3}+
\frac{1}{2}\,{x}^{2}-\frac{1}{2} \right) ^{2}+ \left( \frac{1}{2}\,{x}^{2}+\frac{1}{2}x-{\frac{5}{16}}
 \right) ^{2} \notag\\
 &~~~~~~~~~~~~+\left( \frac{1}{2}x-{\frac{7}{16}} \right) ^{2}+{\frac{1}{128}}.
&
\end{align}
Using YAMILP to compute the SOS for $x^{10}-x+1$, we can obtain the following result: 
$$
{x}^{10}-x+1=f_1^2+f_2^2+\ldots+f_6^2,
$$
where
\begin{align*}
&f_1 = -0.9055741018 + 0.7080  x + 0.5391  x^2 - 0.0840  x^3 + 0.1377  x^4 - 0.2505  x^5,\\
&f_2 = -0.0396579101 - 0.0059  x + 0.2270  x^2 - 0.7687  x^3 - 0.0653  x^4 + 0.8370  x^5,\\
&f_3 = 0.01254393117 - 0.3718  x + 0.6945  x^2 + 0.2441  x^3 - 0.5994  x^4 - 0.0130  x^5,\\
&f_4 = -0.3443097162 - 0.4303  x + 0.0828  x^2 + 0.3887  x^3 + 0.5062  x^4 + 0.3547  x^5,\\
&f_5 = 0.1529814933 - 0.2219  x + 0.2666  x^2 - 0.3059  x^3 + 0.3321  x^4 - 0.3216  x^5,\\
&f_6 = -0.190402017 - 0.1651  x - 0.1344  x^2 - 0.1129  x^3 - 0.1014  x^4 - 0.0853  x^5.
\end{align*}
\smallskip

{\bf Example 7. } The polynomial $p(x)={x}^{6}-2\,{x}^{5}+5{x}^{2}-4\,x+1$
can be expressed by the following SOS forms:
\begin{align}
&{x}^{6}-2\,{x}^{5}+5\,{x}^{2}-4\,x+1\nonumber\\
=& \left( {x}^{3}-{x}^{2}-x+\frac{1}{2}
 \right) ^{2}+ \left( {x}^{2}-\frac{3}{2}\,x+\frac{1}{4} \right) ^{2}+ \left( \frac{3}{2}\,x-\frac{3}{4} \right) ^{2}+\frac{1}{8}
 \label{eq-poly-yang-02}
\\
=& \left( {x}^{3}-{x}^{2}-{\frac {
37\,x}{25}}+{\frac{9}{10}} \right) ^{2}+ \left( 7/5\,{x}^{2}-{\frac {
17\,x}{10}}+{\frac{2699}{7000}} \right) ^{2}\nonumber\\
&+ \left( 4/5\,x-{\frac{877
}{56000}} \right) ^{2}+{\frac{128856407}{3136000000}}.
\end{align}
Below we show that $p(x)$ has no SOS representation with $C=0$ in its lower-triangular matrix.
Assume that $p(x)=XLL^TX^T$, where 
$$
L=\left(\begin{array}{cccc}
1&&&\\
a_{3,2}&s&&\\
a_{3,1}&0&t&\\
a_{3,0}&a_{2,0}&a_{1,0}&a_{0,0}
\end{array}\right).
$$
Let 
$$
L=\left(a_{3,3},a_{3,2},a_{3,1},a_{3,0},a_{2,0},a_{1,0},a_{0,0}\right),
$$
and
$$
D=\left(a_{2,2}=s,a_{1,1}=t\right),
\quad
V=\{a_{2,1}=0\}.
$$
Then, we can construct the following inequality:
\begin{align*}
{\Qo}(s,t) &= (4\,{s}^{10}+ \left( {t}^{2}+16 \right) {s}^{8}+ \left( 28\,{t}^{2}+
88 \right) {s}^{6}+ \left( 8\,{t}^{4}+38\,{t}^{2}+144 \right) {s}^{4} \\
& + \left( 48\,{t}^{4}-228\,{t}^{2}+324 \right) {s}^{2}+(16\,{t}^{6}-120
\,{t}^{4}+225\,{t}^{2}))/(-64s^2t^2)> 0.
\end{align*}
Notice that for for all $s,t\in \mathbb{R}$, 
$$
48\,{t}^{4}-228\,{t}^{2}+324>0, \quad 16\,{t}^{6}-120\,{t}^{4}+225\,{t}^{2}={t}^{2} \left( 4\,{t}^{2}-15
 \right) ^{2}\ge 0,
$$
and therefore ${\Qo}(s,t)> 0$ has no solution. which 
means that $p(x)$ has no SOS representation in form
$$
\left(x^4+a_{3,2}x^2+a_{3,1}x+a_{3,0}\right)^2
+\left(s\cdot x^2+a_{2,0}\right)^2
+\left(t\cdot x+a_{1,0}\right)^2
+a_{0,0}^2
$$
with $s,t,a_{0,0}>0$. 
Following is the SOS representation obtained by YAMILP: 
$$
{x}^{6}-2\,{x}^{5}+5\,{x}^{2}-4\,x+1
=f_1^2+f_2^2+f_3^2+f_4^2,
$$
where
\begin{align*}
&f1 = -0.9043450932 + 2.2666  x - 0.0891  x^2 - 0.5546  x^3,\\
&f2 = 0.07883288347 + 0.1529  x - 1.5081  x^2 + 0.7385  x^3,\\
&f3 = 0.3969217596 + 0.0644  x - 0.1490  x^2 - 0.3600  x^3,\\
&f4 = 0.1356408701 + 0.0896  x + 0.0809  x^2 + 0.1321  x^3.
\end{align*}

{\bf Example 8. } Consider the following problem:
\begin{align*}
p(x)=&\;2\,{x}^{16}-4\,{x}^{15}+6\,{x}^{14}-4\,{x}^{13}+2\,{x}^{12}\\
&\;-{x}^{5}+7
\,{x}^{4}-13\,{x}^{3}+17\,{x}^{2}-11\,x+5
\end{align*}
Notice that 
$$
p(x)=\left( 2\,{x}^{12}-x+5 \right)  \left( {x}^{2}-x+1 \right) ^{2},
$$
and for $2x^{12} - x + 5$,
the lower triangle $L$ contains $2\times 6+1$ border variables in $B$, $1+2+3+4=10$ core variables in $C$, and $5$ diagonal variables in $D$, the inequality system 
$$D>0,\quad  {\Qo}(C,D)>0,$$ 
involves $10+5=15$ variables. First we set $C=0$ and search $D$ in the following grids:
\[
{\mathfrak G}_1 = \left\{ \frac{1}{2},\; 1,\; \frac{3}{2},\; 2 \right\}^5.
\]
Since there are only $4^5=1,024$ points in this grid, it is easy to check the inequality ${\Qo}(C=0,D)>0$ over all points. We found the following two solutions in this grid:
\[
D_1 = \left( \frac{1}{2},\; \frac{1}{2},\; 1,\; \frac{1}{2},\; \frac{1}{2} \right), \quad
D_2 = \left( 1,\; \frac{1}{2},\; 1,\; \frac{1}{2},\; \frac{1}{2} \right),
\]
which yield the following two DDP SOS representations: 

\begin{align}
&
~~~~{x}^{12}-x/2+5/2\notag\\
&~~~~~~~~= \left( {x}^{6}-1/8\,{x}^{4}-{\frac {17\,{x}^{2}}{128
}}-{\frac{529}{1024}} \right) ^{2}+1/4\,{x}^{10}+ \left( 1/2\,{x}^{4}-
{\frac{6501}{16384}} \right) ^{2}\notag\\
&~~~~~~~~~~~~~~+{x}^{6}+ \left( 1/2\,{x}^{2}-{\frac{
25377}{65536}} \right) ^{2}+ \left( x/2-1/2 \right) ^{2}+{\frac{
7197247535}{4294967296}},\label{eq-poly-yang-01}
\\
&~~~~~~~~= \left( {x}^{6}-1/2\,{x}^{4}-1/4\,{x}^{2}-5/8
 \right) ^{2}+{x}^{10}+ \left( 1/2\,{x}^{4}-{\frac{15}{16}}
 \right) ^{2}\notag\\
 &~~~~~~~~~~~~~~+{x}^{6}+ \left( 1/2\,{x}^{2}-{\frac{9}{16}}
 \right) ^{2}+ \left( x/2-1/2 \right) ^{2}+{\frac{85}{128}}.\label{eq-poly-yang-02a}
\end{align}
The expression \eqref{eq-poly-yang-02a} corresponds to the following lower-triangular matrix with zero core variables:
$$
L=\left(\begin{array}{ccccccc}
{\color{red}1}&&&&&&\\
{\color{red}0}&{\color{blue}1}&&&&\\
{\color{red}-1/2}&0&{\color{blue}1/2}&&&\\
{\color{red}0}&0&0&{\color{blue}1}&&\\
{\color{red}-1/4}&0&0&0&{\color{blue}1/2}&&\\
{\color{red}0}&0&0&0&0&{\color{blue}1/2}&\\
{\color{red}-5/8}&{\color{red}0}&{\color{red}-15/16}&{\color{red}0}&{\color{red}-9/16}&{\color{red}-1/2}&{\color{red}85/128}
\end{array}\right)
$$

Without assuming that $C=0$, we further seek solutions of ${\Qo}(C,D)>0$ from the larger grid \( \mathfrak{G}_2\subset \mathbb{R}^{5+10}\): 
\[
{\mathfrak G}_2 = \left\{ \frac{1}{2},\; 1,\; \frac{3}{2},\; 2 \right\}^5 \times
\left\{ 0,\; \frac{1}{2},\; -\frac{1}{2},\; 1,\; -1 \right\}^{10},
\]
This grid contains \( 4^5 \times 5^{10} = 10^{10} \) points. We have used the Monte Carlo method to test ${\Qo}(C,D)>0$ and obtained the following two solutions:
\[
(D_3;V_3) = \left( 1,\; 1,\; 1,\; \frac{1}{2},\; \frac{1}{2};\;
\frac{1}{2},\; 0,\; -1,\; -1,\; 1,\; 0,\; -1,\; 1,\; -\frac{1}{2},\; \frac{1}{2} \right),
\]
\[
(D_4;V_4) = \left( 1,\; \frac{1}{2},\; \frac{1}{2},\; 1,\; 1;\;
0,\; 0,\; \frac{1}{2},\; -1,\; -1,\; 0,\; \frac{1}{2},\; 0,\; 0,\; \frac{1}{2} \right),
\]
and finally got the following two SOS representations:
\begin{align}
x^{12}-\frac{1}{2}x+\frac{5}{2}=&\, \left( {x}^{6}-1/2\,{x}^{4}-1/2\,{x}^{3}-3/4\,{x}^{2}-x/4 \right)
^{2}\notag \\
&+\left( {x}^{5}+1/2\,{x}^{4}-{x}^{2}-x \right) ^{2}
+\left( {x}^{4}+{x}^{3}-x-1/32 \right) ^{2}
\notag\\
&+\left( {x}^{3}+{x}^{2}
-x/2-{\frac {29}{32}} \right) ^{2}
+\left( 1/2\,{x}^{2}+x/2-1
 \right) ^{2}\notag\\
& +\left( x/2-{\frac {15}{32}} \right) ^{2}+{\frac {469
}{1024}}
\label{eq317}\\
=&
\left( {x}^{6}-\frac{1}{2}\,{x}^{4}-\frac{1}{4}\,{x}^{2}+\frac{1}{4} \right) ^{2}
 + \left( {x}^{5}+\frac{1}{2}\,{x}^{2}-x-\frac{1}{4} \right) ^{2}
 \notag\\
&+\left( \frac{1}{2}\,{x}^{4}-{x}^{3}+\frac{1}{2}x-\frac{1}{16} \right) ^{2}+\left( \frac{1}{2}\,{x}^{3}-\frac{1}{8} \right) ^{2}
\notag\\
&+\left( x+\frac{1}{16} \right) ^{2}+ \left( {x}^{2}+\frac{1}{2}x-{\frac {17}{16}}
 \right) ^{2}+{\frac {313}{256}}.
 \label{2x12x5}
\end{align}
From~\eqref{eq317} we have the following SOS representation of $p(x)$: 
\begin{align}
&p(x)=2\, \left( {x}^{8}-{x}^{7}+1
/2\,{x}^{6}-3/4\,{x}^{4}-1/2\,{x}^{2}-x/4 \right) ^{2}\notag\\
&~~~~~~~~+2\, \left( {x}^
{7}-1/2\,{x}^{6}+1/2\,{x}^{5}-1/2\,{x}^{4}-x \right) ^{2}\notag\\
&~~~~~~~~+2\, \left( {
x}^{6}+{\frac {31\,{x}^{2}}{32}}-{\frac {31\,x}{32}}-1/32 \right) ^{2}\notag\\
&~~~~~~~~
+2\, \left( {x}^{5}-1/2\,{x}^{3}+{\frac {19\,{x}^{2}}{32}}+{\frac {13
\,x}{32}}-{\frac{29}{32}} \right) ^{2}\notag\\
&~~~~~~~~+2\, \left( 1/2\,{x}^{4}-{x}^{2}
+3/2\,x-1 \right) ^{2}\notag\\
&~~~~~~~~+2\, \left( 1/2\,{x}^{3}-{\frac {31\,{x}^{2}}{32
}}+{\frac {31\,x}{32}}-{\frac{15}{32}} \right) ^{2}\notag\\
&~~~~~~~~+{\frac {469\,
 }{512}\left( {x}^{2}-x+1 \right) ^{2}}.
\label{eq-poly-yang-02b}
\end{align}

{\bf Example 9. } 
Consider the following polynomial:
$$
p(x)=2{x}^{16}-4{x}^{15}-2{x}^{14}+4{x}^{13}+2{x}^{12}-{x}^{5}+7
{x}^{4}-9{x}^{3}-7{x}^{2}+9x+6.
$$
In general, if $x=x_0$ is a minimum point of $p(x)$ and $p(x)$, then
$$
(x-x_0)|\gcd(p(x),p'(x)),
$$
and 
$$
p(x)=\left(\gcd(p(x),p'(x)\right)^2\, q(x)+p(x_0).
$$
For $p(x)$ given in this example, we have
$$
\gcd(p'(x),p(x)=x^2-x-1,
$$
therefore
$$
p(x)=(x^2 - x -1)^2( 2x^{12} - x + 5)+1,
$$
Recall that we have constructed 4 different SOS representations of $2x^{12} - x + 5$ in previous example. Using the result \eqref{2x12x5}, we have

\begin{align}\label{eq-poly-yang-05}
&
~~~~
2{x}^{16}-4{x}^{15}-2{x}^{14}+4{x}^{13}+2{x}^{12}-{x}^{5}+7
{x}^{4}-9{x}^{3}-7{x}^{2}+9x+6\notag\\
&~~~~~~~~=2 \left( {x}^{8}-{x}^{7}-\frac{3}{2}
{x}^{6}+\frac{1}{2}{x}^{5}+\frac{1}{4}{x}^{4}+\frac{1}{4}{x}^{3}+\frac{1}{2}{x}^{2}-\frac{1}{4} x-\frac{1}{4}
 \right) ^{2}\notag\\
&~~~~~~~~\phantom{=}+2 \left( {x}^{7}-{x}^{6}-{x}^{5}+\frac{1}{2}{x}^{4}-\frac{3}{2}{x}
^{3}+\frac{1}{4}{x}^{2}+\frac{5}{4}x+\frac{1}{4} \right) ^{2}\notag\\
&~~~~~~~~\phantom{=}+2 \left( \frac{1}{2}{x}^{6}-\frac{3}{2}
{x}^{5}+\frac{1}{2}{x}^{4}+\frac{3}{2}{x}^{3}-{\frac {9}{16}{x}^{2}}-{\frac {7
}{16}x}+\frac{1}{16} \right) ^{2}\notag\\
&~~~~~~~~\phantom{=}+2 \left( \frac{1}{2}{x}^{5}-\frac{1}{2}{x}^{4}-\frac{1}{2}
{x}^{3}-\frac{1}{8}{x}^{2}+\frac{1}{8} x+\frac{1}{8} \right) ^{2}\notag\\
&~~~~~~~~\phantom{=}+2 \left( {x}^{4}-\frac{1}{2}{x}^
{3}-{\frac {41}{16}{x}^{2}}+{\frac {9}{16}x}+{\frac{17}{16}}
 \right) ^{2}
 \notag\\&~~~~~~~~\phantom{=}
 +2 \left( {x}^{3}-{\frac {15}{16}{x}^{2}}-{\frac {17
}{16}x}-\frac{1}{16} \right) ^{2}
+{\frac {313 }{128}\left( {x}^{2}-x-1 \right) ^
{2}}
+1. 
&
\end{align}
ÏUsing YAMILP, we obtained the following SOS representation: 
$$
p(x)=f_1^2+f_2^2+\cdots+f_9^2,
$$
where
\begin{align*}
&f_1 = - 2.282003164 - 1.8515  x+ 2.9548  x^2 + 0.3730  x^3- 0.1631  x^4\\
&\phantom{xxxx} + 0.4538  x^5 - 0.3186  x^6 + 0.1180  x^7 - 0.1230  x^8,\\
&f_2 = 0.3504639324 + 0.2600  x + 0.1273  x^2 + 0.5308  x^3+ 0.3520  x^4\\
&\phantom{xxxx} + 0.8840  x^5- 2:1220  x^6 - 1.3088  x^7 + 1.2890  x^8,\\
&f_3 = -0.01057227405 - 0.2774  x - 0.1471  x^2 + 1.5611  x^3+ 0.7024  x^4\\
&\phantom{xxxx} - 1.5902  x^5- 0.3720  x^6 + 0.3534  x^7 + 0.0756  x^8,
\end{align*}
\begin{align*}
&f_4 = 0.392835776 + 0.3929  x+ 0.4646  x^2 + 0.8757  x^3- 1.8251  x^4\\
&\phantom{xxxx} - 0.2674  x^5 + 0.4692  x^6 - 0.5102  x^7 + 0.3436  x^8,\\
&f_5 = -0.2120484475 - 0.9874  x- 0.9281  x^2 + 0.9047  x^3+ 0.1473  x^4\\
&\phantom{xxxx}+ 0.8984  x^5 + 0.7269  x^6 - 0.6242  x^7 - 0.1176  x^8,\\
&f_6 = 0.1701900667 - 0.5797  x- 0.3835  x^2 - 0.0336  x^3- 0.6067  x^4\\
&\phantom{xxxx} + 0.1371  x^5- 0.6304  x^6 + 0.7216  x^7 - 0.1110  x^8,\\
&f_7 = -0.4560569811 - 0.2364  x - 0.4610  x^5 - 0.2038  x^4 - 0.4060  x^3\\
&\phantom{xxxx}- 0.3606  x^2 - 0.0081  x^6 - 0.3210  x^7 + 0.4070  x^8,\\
&f_8 = -0.4829986477 + 0.3978  x- 0.1909  x^2 + 0.1727  x^3- 0.0856  x^4\\
&\phantom{xxxx}+ 0.1339  x^5 - 0.0990  x^6 + 0.1427  x^7 - 0.0877  x^8,\\
&f_9 = 8.341098855\times 10^{- 5} + 3.8276\times 10^{- 4}  x + 4.8301\times 10^{- 4}  x^2 \\
&\phantom{xxxx}+ 8.8392\times 10^{- 4}  x^3+ 0.0014  x^4+ 0.0023  x^5+ 0.0037  x^6 \\
&\phantom{xxxx}+ 0.0061  x^7 + 0.0099  x^8.
\end{align*}

{\bf Example 10. } Express the following polynomial in SOS form: 
\begin{align*}
p(x)=&{x}^{28}+{x}^{27}+{x}^{26}+{x}^{25}+{x}^{24}+{x}^{23}+{x}^{22}+{x}^{21
}+{x}^{20}+{x}^{19}
\\
&
+{x}^{18}+{x}^{17}+{x}^{16}+{x}^{15}+{x}^{14}+{x}^{
13}+{x}^{12}+{x}^{11}+{x}^{10}
\\
&+{x}^{9}+{x}^{8}+{x}^{7}+{x}^{6}+{x}^{5}
+{x}^{4}+{x}^{3}+{x}^{2}+x+1.
\end{align*}

For this problem. we observe the following SOS representations:
\begin{align*}
&\phantom{xxx}x^4+x^3+\ldots+x+1=
(x^2+\frac{1}{2}x)^2+\frac{3}{4}(x+\frac{2}{3})^2+\frac{2}{3},\\[3pt]
&\phantom{xxx}x^6+x^5+\ldots+x+1=x^2(x^4+x^3+\ldots+x+1)+x^2+x+1\\[3pt]
&
=(x^3+\frac{1}{2}x^2)^2+\frac{3}{4}(x^2+\frac{2}{3}x)^2+\frac{2}{3}(x+\frac{3}{4})^2+\frac{5}{8},\\[3pt]
&\phantom{xxx}x^8+x^7+\ldots+x+1=x^2(x^6+x^5+\ldots+x+1)+x^2+x+1\\[3pt]
&\;=(x^4+\frac{1}{2}x^3)^2+\frac{3}{4}(x^3+\frac{2}{3}x^2)^2+\frac{2}{3}(x^2+\frac{3}{4}x)^2
+\frac{5}{8}(x+\frac{4}{5})+\frac{3}{5}.
\end{align*}
Therefore, for this polynomial, we try to search lower triangular matrix $L$ in the following form
\begin{equation}
L=\left(\begin{array}{cccccccc}
a_n&\\
b_n&a_{n-1}&\\
0&b_{n-1}&a_{n-2}&\\
0&0&b_{n-2}&a_{n-3}&\\
\vdots&\vdots&\vdots&\vdots&\ddots&\\
0&0&0&0&\cdots&a_1&\\
0&0&0&0&\cdots&b_1&\;\,a_0
\end{array}\right)\quad (n=14),
\label{eqs1-3}
\end{equation}
to make 
$$
p(x)=(x^{14},\ldots,x,1)L\, L^T(x^{14},\ldots,x,1).
$$
So we get $2n+1$ equations,   

\begin{align*}
&\,\left(a_{i,i}|_{i=28,\ldots,2,1}\right)
=\left(1,\,\sqrt{\frac{3}{4}},\,\sqrt{\frac{4}{6}},\,\sqrt{\frac{5}{8}},\,\sqrt{\frac{6}{10}},\,\sqrt{\frac{7}{12}},\,
\sqrt{\frac{8}{14}},\;\ldots,\; \sqrt{\frac{16}{30}},
\right),\\
&\,\left(a_{i,i-1|i=28,\ldots,2}\right)
=\left(\sqrt{\frac{1}{4}},\,\sqrt{\frac{1}{6}},\,\sqrt{\frac{1}{8}},\,\sqrt{\frac{1}{10}},\,
\sqrt{\frac{1}{12}},\,\sqrt{\frac{1}{14}},\,
\ldots,\; \sqrt{\frac{1}{30}}\right).
\end{align*}
Thus, 
\begin{align}
p(x)=&
\left( {x}^{14}+\frac{1}{2}\,{x}^{13} \right) ^{2}
+\frac{3}{4}\, \left( {x}^{13}+\frac{2}{3}\,{x}^{12} \right) ^{2}
+\frac{2}{3}\, \left( {x}^{12}+\frac{3}{4}\,{x}^{11} \right) ^{2}
\nonumber\\
&
+\frac{5}{8}\, \left( {x}^{11}+\frac{4}{5}\,{x}^{10} \right) ^{2} +\frac{3}{5}\, \left( {x}^{10}+\frac{5}{6}\,{x}^{9} \right) ^{2}
 +{\frac {7}{12}\, \left( {x}^{9}+\frac{6}{7}\,{x}^{8} \right) ^{2}}
\nonumber\\
&
+\frac{4}{7}\, \left( {x}^{8}+{\frac {7}{8}\,{x}^{7}} \right) ^{2}
+{\frac {9}{16} \left( {x}^{7}+{\frac {8}{9}\,{x}^{6}} \right) ^{2}}+\frac{5}{9}\, \left( {x}^{6}+{\frac {9}{10}\,{x}^{5}} \right) ^{2}\nonumber
\\
&
+{\frac {11}{20} \left( {x}^{5}+{\frac {10}{11}\,{x}^{4}} \right) ^{2}}
+{\frac {6}{11} \left( {x}^{4}+{\frac {11}{12}\,{x}^{3}} \right) ^{2}}
\nonumber\\
&
+{\frac {13}{24} \left( {x}^{3}+{\frac {12}{13}\,{x}^{2}} \right) ^{2}}
+{\frac {7}{13} \left( {x}^{2}+{\frac {13}{14}\,x} \right) ^{2}}
\nonumber\\
&
+{\frac {15}{28} \left( x+{\frac{14}{15}} \right) ^{2}}
+{\frac{8}{15}}.
\label{poly28}
\end{align}

\section{Rational SOS for multi-variate polynomials}

A ``project and lifting'' technique used in factorization of a multivariate polynomial \( F(x_1, x_2, \ldots, x_n) \) is as follows: in the {\it project\/} step we take a sequence of sufficiently large natural numbers 
\[ k_1 < k_2 < \ldots < k_n, \]
substitute the following transformation rule
\begin{equation}
x_1 = t^{k_1}, \; x_2 = t^{k_2}, \; \ldots, \; x_n = t^{k_n},
\label{substt}
\end{equation}
into $F(x_1,x_2,\ldots,x_n)$, and factorize
the obtained univariate polynomial
$$
G(t)=F(t^{k_1},t^{k_2},\ldots,t^{l_n})
\longrightarrow 
g_1(t) \times \ldots \times g_m(t).
$$
In the lifting step, we transform each polynomial $g_i(t)$ to a multi-variate polynomial $f_i(x_1,x_2,\ldots,x_n)$ 
through the following {\it successive remainder computation\/}: 
\begin{align}
   g_i(t)=r_0&\;{}{\longrightarrow}
   r_1(x_n;t):=\text{rem}(r_0,t-x_n^{k_n}, t),\nonumber\\
   r_1&\;{}{\longrightarrow}
   r_2(x_{n-1},x_{n};t):=\text{rem}(r_1,t-x_{n-1}^{k_{n-1}}, t),\nonumber\\
   r_2&\;{\longrightarrow}
   r_3(x_{n-2},x_{n-1},x_{n};t):=\text{rem}(r_1,t-x_{n-2}^{k_{n-2}}, t),\nonumber\\
   &\;\vdots\nonumber\\
   r_{n-1}&\;{}{\longrightarrow}
   r_{n}(x_1,\ldots,x_{n-1},x_n):=\text{rem}
   (r_{n-1},t-x_1^{k_1},t).
   \nonumber\\
  &\phantom{{}\longrightarrow}f_i:=r_n(x_1,x_2,\ldots,x_n).
  \label{eq-src}
\end{align}
Finally, we get the factorization of $F$:
$$
F(x_1,x_2,\ldots,x_n)=
\prod_{i=1}^{m}f_i(x_1,x_2,\ldots,x_n).
$$

This technique can also be applied to reconstruct the SOS representation of multivariate positive definite polynomials, provide we know already that this polynomial can be expressed in SOS forms.
The main process is as follows: 

Let 
\( p(x_1, x_2, \ldots, x_n) \) be a  multivariate, $k_1<k_2<\ldots<k_n$ be the sequence of natural numbers generated by the following procedure:
\begin{align*}
&k_1=1, \\   
&k_2=1+\deg(p(x_1=t^{k_1}),t)\\
&k_3=1+\deg(p(x_1=t^{k_1},x_2=t^{k_2}),t),\\
&\vdots\\
&k_n=1+\deg(p(x_1=t^{k_1},x_2=t^{k_2},\ldots, x_{n-1}=t^{k_{n-1}}),t), 
\end{align*}
and 
$$
G(t)=p(x_1=t^{k_1},x_2=t^{k_2},\ldots, x_{n-1}=t^{k_{n-1}}, x_n=t^{k_n}).
$$

Clearly, the projected polynomial $G(t)$ is positive definite if $p(x_1,x_2,\ldots,x_n)$ is positive definite. 
and therefore we can write $G(t)$ as a sum of square of certain polynomials: 
\[
G(t) = g_N(t)^2 + g_{N-1}(t)^2 + \ldots + g_1(t)^2 + g_0.
\]
Let $q_0=g_0$, and $q_1,q_2,\ldots,q_N\in \mathbb{R}[x_1,x_2,\ldots,x_n]$ be multivariate polynomials obtained by the previously mentioned successive remainder computation from $g_1(t),g_2(t),\ldots,g_N(t)$, respectively. Then,
\begin{equation}
p=q_N(x_1,x_2,\ldots,x_n)^2+q_{N-1}(x_1,x_2,\ldots,x_n)^2+\cdots+q_1(x_1,x_2,\ldots,x_n)^2+q_0^2
\label{eq-proj-lift}    
\end{equation}

In general, $N$ is much larger than the degree of $p(x_1,x_2,\ldots,x_n)$, which brings a very high computational complexity for solving the equation:
$$
G(t)=(t^N,t^{N-1},t^{N-2},\ldots,t^2,t,1)LL^T(t^N,t^{N-1},t^{N-2},\ldots,t^2,t,1)^T.
$$
To solve this difficulty, we can try to find
integers $s_j\,(j=1,2,\ldots,l)$ so that  
$$
0\leq {s_1}<{s_2}<\cdots< s_{l-1}>{s_l}<N
$$
and try to find a lower-triangular matrix $L$
so that 
$$
G(t)=
(t^{N}, t^{s_l},t^{s_{l-1}}, \ldots,t^{s_2},t^{s_{1}})
LL^T
(t^{N}, t^{s_l},t^{s_{l-1}}, \ldots,t^{s_2},t^{s_{1}})^T.
$$
Clearly, if this equation has a solution, then we can reconstruct an SOS of $G(t)$ with less computation, and therefore finally construct an SOS representation of $p(x_1,x_2,\ldots,x_n)$ using the successive remainder computation. 

The following  {\it power selection procedure\/} can be used to construct a such sequence $(t^{N}, t^{s_l},t^{s_{l-1}}, \ldots,t^{s_2},t^{s_{1}})$:

{\tt Input:} a multivariate polynomial $p(x_1,x_2,\ldots,x_n)$. 

{\tt Initialize step:}  Let
$$
{\cal N}=\emptyset,
\quad 
{\cal M}=\{(d_1,d_2,\ldots,d_n): \text{coeff}(p,x_1^{d_1}x_2^{d_2}\cdots x_n^{d_n})\not=0\}, 
$$
{\tt for} $(d_1,d_2,\ldots,d_n)\in {\cal M}$ {\tt do}: 
if $\mod(d_i,2)=0$ for $i=1,2,\ldots,n$ then 
$$
{\cal N}\leftarrow (d_1/2,d_2/2,\ldots,d_n/2),\quad
{\cal M}={\cal M}\setminus\{(d_1,d_2,\ldots,d_n)\}.
$$
{\tt end if, end do.}

{\tt Loop: } {\tt while} ${\cal M}\not=\emptyset$ {\tt for} $m=(d_1,d_2,\ldots,d_n)\in {\cal M}$ {\tt do}: {\tt for} $j$ from 1 to $\#{\cal N}$ {\tt do} 
$$
m_j:=(e_1,e_2,\ldots,e),
\text{ the $j$-th element of } {\cal N}, 
$$
a{\tt if }
$$
d_i\geq e_i \text{ for all } i=1,2,\ldots,n \text{ and }
d_1+d_2+\ldots+d_n>e_1+e_2+\ldots+e_n,
$$
then
$$
{\cal N}\leftarrow 
(d_1-e_1,d_2-e_2,\ldots,d_n-e_n), 
\quad
$$
{\tt end if}, {\tt end do}, 
$$
{\cal M}={\cal M}\setminus \{m\}.
$$
{\tt end do.}

{\tt output: }
$$
{\cal S}:=\left\{ {k_1l_1+k_2l_2+\ldots+k_nl_n},\; (l_1,l_2,\ldots,l_n)\in {\cal N} \right\}.
$$

Without loss of generality, we may assume that ${\cal S}$ contains $l+1$ integers and they are arranged as follows:
$$
s_1<s_2<\cdots<s_l<N=s_{l+1}. 
$$
Let 
$$
{\cal T}=(t^N,t^{s_l},\ldots,t^{s_2},t^{s_1}).
$$
Then we can apply the algorithm we have presented in Section~\ref{algorithm} to construct the SOS in the following form:
\begin{equation}\label{eq-02-0005}
q(t)=
(t^N,t^{s_l},\ldots,t^{s_1})HH^T(t^N,t^{s_l},\ldots,t^{s_1})
^T
=\sum_{j=1}^{l+1}\left(\sum_{i=1}^{j}a_{j,i}t^{s_i}\right)^2,
\end{equation}
and finally, reconstruct $q_(x_1,x_2,\ldots,x_n)$ by performing the 
successive remainder computation on the univariate polynomial  
$$
g_j=\sum_{i=1}^{j}a_{j,i}t^{s_j}, \quad j=1,2,\ldots,l+1. 
$$
The proof of the existence of the SOS in \eqref{eq-02-0005} is rather complicated. Here we just show a few examples.
\medskip 

{\bf Example 11. } The following polynomial came from Find \cite{Powers1998}. 
$$
f(x,y,z)={x}^{4}+{x}^{3}z+2\,{x}^{2}{y}^{2}+{z}^{4}.
$$
Using YAMILP, this polynomial can be expressed
in the following SOS form:
$$
f(x,y,z)
=f_1^2+f_2^2+f_3^2+f_4^2+f_5^2,
$$
where
\begin{align*}
&f_1 = 1.4142  x  y,\\
&f_2 = 0.9238  x^2 - 0.6473  z^2 + 0.6329  x  z,\\
&f_3 = -0.0406  x^2 - 0.7182  z^2 - 0.6752  x  z,\\
&f_4 = 0.3807  x^2 + 0.2554  z^2 - 0.2945  x  z,\\
&f_5 = 8.9800\times 10^{ - 7}  y  z.
\end{align*}
According to this result, we guess that
$f(x,y,z)$ can be expressed as
$$
f(x,y,z)=2x^2y^2+q_2(x,z)^2+q_3(x,z)^2+q_4(x,z)^2.
$$
Below we show the main procedures to reconstruct 
an SOS representation of $f(x,y,z)$ with rational coefficients.

In the {\it univariate polynomial projection} step we decide
$$
k_1=1, \; k_2=4, \; k_3=13
$$
and
$$
G(t)=p(t,t^5,t^{13})=
(t,t^5,t^{13})={t}^{52}+{t}^{16}+2\,{t}^{12}+{t}^{4}. 
$$
In the {\it power selection procedure\/} we have
\begin{align*}
    {\cal M}=&\;\{(4,0,0),(3,0,1),(2,2,0),(0,0,4)\},\\
    {\cal N}\leftarrow &\;\{(2,0,0),(1,1,0),(0,0,2)\},\\
    {\cal M}=&\;\{(3,0,1)\},\\
    {\cal N}\leftarrow&\; \{(1,0,1)\},\\
    M=&\;\emptyset,\\
    {\cal N}=&\;\{(2,0,0),(1,1,0),(0,0,2),(1,0,1)\},
\end{align*}
and therefore 
$$
{\cal T}=\text{subst}(\{x=t,y=t^5,z=t^{13}\},\,
\{{x}^{2},yx,xz,{z}^{2}\})=\{{t}^{2},{t}^{6},{t}^{14},{t}^{26}\}.
$$
Thus we search DDP SOS form of the following form:
\begin{align}
  G(t) &= \left( a_{{3}}{t}^{26}+a_{{2}}{t}^{14}+a_{{1}}{t}^{6}+a_{{0}}{t}^{2}
 \right) ^{2}\notag\\
 &+ \left( b_{{2}}{t}^{14}+b_{{1}}{t}^{6}+b_{{0}}{t}^{2}
 \right) ^{2}\notag\\
 &+ \left( c_{{1}}{t}^{6}+c_{{0}}{t}^{2} \right) ^{2}\notag\\
  &+({d_{{0}}}{t}^{2})^2.\notag
\end{align}
The result of this step is
$$
a_3=1,\; a_2=0,\; a_1=0,\; a_0=-1/2,
$$
$$
b_2=1,\; b_1=0,\; b_0=1/2,
$$
$$
c_1=\sqrt{2},\; c_0=0,\; d_0=\sqrt{2}/2,
$$
which leads to the following SOS form:
\begin{equation}\label{eq-000013}
{t}^{52}+{t}^{16}+2\,{t}^{12}+{t}^{4}=\left( {t}^{26}-\frac{1}{2}\,{t}^{2} \right) ^{2}+ \left( {t}^{14}+\frac{1}{2}\,{t}^{
2} \right) ^{2}+2\,{t}^{12}+\frac{1}{2}\,{t}^{4}.
\end{equation}
Finally, by performing the following {\it successive remainder computing\/}: 
$$
\mathrm{rem}(~\cdot~,t^{13}-z,t),~\mathrm{rem}(~\cdot~,t^{5}-y,t), ~\mathrm{rem}(~\cdot~,t-x,t).
$$
we get obtain the desired SOS representation.
\begin{equation}\label{eq-000015}
{x}^{4}+{x}^{3}z+2\,{x}^{2}{y}^{2}+{z}^{4}=\left( {z}^{2}-\frac{1}{2}\,{x}^{2}
 \right) ^{2}+ \left( xz+\frac{1}{2}\,{x}^{2} \right) ^{2}+2\,{x}^{2}{y}^
{2}+ \frac{1}{2}\,{x}^{4}.
\end{equation}

{\bf Example 12.} Consider the following bi-variate polynomial: 
$$
p(x,y)=x^6+2x^5y+5x^2y^4+4xy^5+y^6.
$$
Since it is a homogeneous polynomial  and can be easily transformed to a univariate polynomial, which guarantees the existence of SOS form expression of $p(x,y)$. Using YAMILP we have the following result: \begin{align*}
&f_1 = -0.9043  y^3 - 2.2666  x  y^2 - 0.0891  x^2  y + 0.5546  x^3,\\
&f_2 = 0.0788  y^3 - 0.1529  x  y^2 - 1.5081  x^2  y - 0.7385  x^3,\\
&f_3 = -0.3969  y^3 + 0.0644  x  y^2 + 0.1490  x^2  y - 0.3600  x^3,\\
&f_4 = -0.1356  y^3 + 0.0896  x  y^2 - 0.0809  x^2  y + 0.1321  x^3.
\end{align*}

In our method, we take $x=t,y=t^7$, and 
$$
G(t)=p(t,t^7)=
t^{42}+4t^{36}+5t^{30}+2t^{12}+t^6.
$$
In the power selection procedure we have
\begin{align*}
    {\cal M}=&\;\{(6,0),(5,1),(2,4),(1,6),(0,6)\},\\
    {\cal N}\leftarrow &\;\{(3,0),(1,2),(0,3)\},\\
    {\cal M}=&\;\{(5,1),(1,5)\},\\
    {\cal N}\leftarrow&\; \{(2,1),(0,3),(1,2)\},\\
    M=&\;\emptyset,\\
    {\cal N}=&\;\{(3,0),(1,2),(0,3),(1,2)\},
\end{align*}
and consequently,
$$
{\cal T}=\text{subst}(\{x=t,y=t^7\}, \{x^3,x^2y,xy^2,y^3\})
=\{t^3,t^9,t^{15},t^{21}\}.
$$
Therefore, we search SOS form of $G(t)$ in the following form:
\begin{align*}
g(t)&=(a_3t^{21}+a_2t^{15}+a_1t^9+a_0t^3)^2\\
&+(b_2t^{15}+b_1t^9+b_0t^3)^2\\
&+(c_1t^9+c_0t^3)^2\\
&+(d_0t^3)^2,
\end{align*}
and obtain the following solution:
\begin{align*}
&a_3=1,\quad a_2=2, \quad a_1=0, \quad a_0=-1/2,\\
&b_2=1,\quad b_1=1/2, \quad b_0=-1/4,\\
&c_1=3/2,\quad c_0=3/4,\\
&d_0^2=1/8.
\end{align*}
This leads the following SOS representation of $G(t)$
as
\begin{align*}
&t^{42}+4t^{36}+5t^{30}+2t^{12}+t^6\\
=&\left(t^{21}+2t^{15}-\frac{1}{2}t^3\right)^2
+\left(t^{15}+\frac{1}{2}t^9-\frac{1}{4}t^3\right)^2
+\left(\frac{3}{2}t^9+\frac{3}{4}t^3\right)^2
+\frac{1}{8}t^6.
\end{align*}
Performing the successive remainder computing:
$$
{\tt rem}(\cdot, t^7-y,t), \quad
{\tt rem}(\cdot, t-x,t),
$$
finally we get:
\begin{align*}
&x^6+2x^5y+5x^2y^4+4xy^5+y^6\\
=&\left(y^{3}+2xy^{2}-\frac{1}{2}x^3\right)^2
+\left(xy^{2}+\frac{1}{2}x^2y-\frac{1}{4}x^3\right)^2
+\left(\frac{3}{2}x^2y+\frac{3}{4}x^3\right)^2
+\frac{1}{8}x^6.
\end{align*}

{\bf Example 13.}
Find an SOS representation of 
\[
f(x, y, z) = x^6 + 4x^3y^2z + y^6 + 2y^4z^2 + y^2z^4 + 4z^6.
\]
In the project step we use the following substitution:
$$
\textit{proj}: x=t, \quad y=t^7, \quad  z=t^{43}.
$$
It maps $p(x,y,z)$ to the following univariate polynomial
\[
G(t) = f(t, t^7, t^{43}) = 4t^{258} + t^{186} + 2t^{114} + 4t^{60} + t^{42} + t^6.
\]
In the power selection procedure, we have
$$
{\cal M}=\{(3,0,,0),(0,3,0),(0,2,1),(0,1,2),(0,0,3)\}, \quad 
{\cal K}=\{(0,2,1)\}.
$$
and
$$
\{x^3, y^3, y^2z, yz^2, z^3\}\stackrel{\textit{proj}}{\longrightarrow} {\cal T}=\{t^3, t^{21}, t^{57}, t^{93}, t^{129}\}.
$$
Further, we search the SOS representation of \( G(t) \) in the following form:
\[
g(t) = \left( a_4 t^{129} + a_3 t^{93} + a_2 t^{57} + a_1 t^{21} + a_0 t^3 \right)^2
+ \left( b_3 t^{93} + b_2 t^{57} + b_1 t^{21} + b_0 t^3 \right)^2
\]
\[
+ \left( c_2 t^{57} + c_1 t^{21} + c_0 t^3 \right)^2
+ \left( d_1 t^{21} + d_0 t^3 \right)^2
+ (e_0 t^3)^2.
\]
Use our algorithm presented in Section~\ref{algorithm} we obtain the following solution:
\[
a_4 = 2, \quad a_3 = 0, \quad a_2 = 0, \quad a_1 = 0, \quad a_0 = 0,
\]
\[
b_3 = 1, \quad b_2 = 0, \quad b_1 = -1, \quad b_0 = 0,
\]
\[
c_2 = 2, \quad c_1 = 0, \quad c_0 = 1,
\]
\[
d_1 = 0, \quad d_0 = 0, \quad e_0 = 0.
\]
This solution derives the following SOS representation of \( G(t) \):
\[
4 t^{258} + t^{186} + 2 t^{114} + 4 t^{60} + t^{42} + t^6 = (2 t^{129})^2 + (t^{93} - t^{21})^2 + (2 t^{57} + t^3)^2,
\]
and consequently, by performing the following successive remainder computation,  
\[
\text{rem}({ ~\cdot~}, t^{43} - z, t), \quad \text{rem}(~\cdot~, t^{7} - y, t), \quad \text{rem}(~\cdot~, t - x, t),
\]
we finally obtain the SOS representation of \( p(x,y,z) \):
\[
x^6 + 4 x^3 y^2 z + y^6 + 2 y^4 z^2 + y^2 z^4 + 4 z^6 = (2 z^3)^2 + (z^2 y - y^3)^2 + (x^3 + 2 y^2 z)^2.
\]

\textbf{Example 14. } 
(Motzkin's counterexample \cite{Motzkin1965}) 
Show the following polynomial
$$ f( {x,y} ) = {x}^{4}{y}^{2} + {x}^{2}{y}^{4} - 3{x}^{2}{y}^{2} + 1 $$
does not admit a sum of squares (SOS) representation, and construct an SOS form of 
$$
(x^2+y^2+1)f(x,y).
$$

Assume that $f(x,y) \ge 0$ admits a sum of squares decomposition. Let
$$
\textit{proj}: x=t,\; y=t^4.
$$
Then $f(x,y)$ is projected to
$$
G(t) = f(t,t^5) = t^{22} + t^{14} - 3 t^{12} + 1 
$$
In power selection procedure, we have
$$
{\cal N}=\{(2,1),(1,2),(1,1),(0,0)\}, 
$$
and
$$
{\cal T}=\{1, t^6, t^7, t^{11}\}.
$$
So if $G(t)$ has any SOS representation, 
it must be the following form:
$$ G(t)= \left( a_3 t^{11} + a_2 t^7 + a_1 t^6 + a_0 \right)^2 + \left( b_2 t^7 + b_1 t^6 + b_0 \right)^2 + \left( c_1 t^6 + c_0 \right)^2 + d_0^2. 
$$
By comparing the coefficient of $t^{12}$, we have the following equation 
$$
a_1^2 + b_1^2 + c_1^2 + 3=0,
$$
which has no real solution, therefore, 
$G(t)$ does not admit a sum of squares, and therefore, $f(x,y) $ does not have a sum of squares representation.

Let 
\begin{align*}
f_1(x,y)=&\;(x^2 + y^2 + 1) f(x,y)\\    
=&\;x^6 y^2 + 2 x^4 y^4 + x^2 y^6 - 2 x^4 y^2 - 2 x^2 y^4 - 3 x^2 y^2 + x^2 + y^2 + 1.
\end{align*}
We aim to construct an SOS for this polynomial. For this problem, we have
$$
x=t,\quad y=t^7,
$$
and
$$
G_1(t)=t^{44}+2t^{32}-2t^{30}+t^{20}-2t^{18} 
-3t^{16}+t^{14}+t^2+1. 
$$
In the power selection procedure we have
$$
{\cal N}=\{(3,1),(2,2),(1,3),(2,1),(1,2),(1,1),(1,0),(0,1),(0,0)\}, 
$$
and
$$
{\cal T} = \{1, t, t^7, t^8, t^9, t^{10}, t^{15}, t^{16}, t^{22}\}.
$$
Let 
$$
V=\left(t^{22},t^{16},t^{15},t^{10},t^{9},t^{8},t^{7},t,1\right).  
$$
Consider the following following equation
$$
G_1(t)=(VL)(VL)^T,
$$
where $L$ is a lower triangular matrix of order $9$:
$$
L=\left(\begin{array}{cccccc}
a_{99}&\\
a_{98}&a_{88}&\\
a_{97}&a_{87}&a_{77}&\\
\vdots&\vdots&\vdots&\ddots&\\
a_{91}&a_{81}&a_{71}&\cdots&a_{11}&\\
a_{90}\;&\;a_{80}\;&\;a_{70}\;&\;\cdots\;&\;a_{10}\;&\;a_{00}
\end{array}
\right).
$$
It derives a system with $9$ polynomial equations, involving $45$ variables. Among the variables, $17$ are in the border $B$, $7$ are diagonal variables, and $21$ are in the core $C$. Following is one solution we found:
$$
L^{\ast}=\left(\begin{array}{ccccccccc}
1&\\
0&1&\\
0&0&1&\\
1/2&0&0&\sqrt{3}/2&\\
0&0&0&0&1&\\
-3/2&0&0&-\sqrt{3}/2&0&0\\
0&0&0&0&1&0&0\\
0&0&-1&0&0&0&0&0\\
0&\;-1\;\;&0&0&\;\;0\;\;&\;0\;&\;\;0\;\;&\;\;0\;\;&\;\;0
\end{array}
\right).
$$
which leads to the SOS expressions:
$$
G_1(t) = \left(t^{22} + \frac{1}{2} t^{10} - \frac{3}{2} t^8 \right)^2 + \left(t^{16} - 1\right)^2
+ \left(t^{15} - t \right)^2 
$$ 
$$
+ \left( \frac{t^{10} \sqrt{3}}{2} - \frac{t^8 \sqrt{3}}{2} \right)^2 + \left(t^9 - t^7\right)^2,
$$
and
$$
f_1(x,y) = \frac{1}{4} \left( 2xy^3 + x^3 y - 3xy \right)^2 + (x^2 y^2 - 1)^2 
$$
$$
+ (xy^2 - x)^2 + \frac{3}{4} (x^3 y - xy)^2 + (x^2 y - y)^2.
$$
The following sum of squares is can seen in Parrilo \cite{Parrilo2003}:
\begin{align*}
(x^2 + y^2 + 1) f(x,y) &= \left( x^2 y - y \right)^2 + \left( xy^2 - x \right)^2 + \left( x^2 y^2 - 1 \right)^2 \\
&+ \frac{1}{4} \left( xy^3 - x^3 y \right)^2 + \frac{3}{4} \left( xy^3 + x^3 y - 2xy \right)^2,
\end{align*}
this form corresponds to the following univariate polynomial under the project 
$$
G_1(t)=(t^8-t^6)^2+(t^{15}-t)^2+(t^{16}-1)^2
$$
$$
+
\frac{1}{4}(t^{22}-t^{10})^2+\frac{3}{4}(t^{22}-t^{10}-2t^{8})^2. 
$$

\section{Conclusion}

The SOS decomposition problem for (semi)-positive semidefinite polynomials is well-known and has been extensively studied by many researchers (see \cite{Motzkin1965,Choi1995,Reznick2000,Rudin2000,Oliveira2006,Harrison2007,Schmuedgen2012}). 
Numerical implementations (such as YAMILP, SeDuMi, and SOSTOOLS) are available in various mathematical software packages (see \cite{Powers1998,Parrilo2003,Parrilo2012,Powers2015}).  
Typically, the focus has been on efficiently solving the factorization problem \( B = HH^T \) in the space of \( (d+1) \times (d+1) \) symmetric matrices. 

A few years ago, while the authors were discussing how to construct the exact SOS representation for several relatively simple univariate polynomials came from in \cite{Feng1982,Feng2020} by hand, we luckily found that it is always possible to generate an SOS of degree-descending polynomials. In most cases, when solving the equations and inequalities generated by these problems, a solution exists even when the core variables (i.e., \( a_{i,j} \) with \( 1 \leq j < i \leq d \)) are set to zero. 
In the past, research on SOS with rational coefficients (see \cite{PeyrlParrilo2007a,PeyrlParrilo2007b,KLYZ2008,Kaltofen2009,KLYZ2012,Menini2015}
)  either focused on finding rational points in the semialgebraic set or on recovering the accurate solution from results with errors. As a result, people generally did not pay attention to the fine structure of the solutions to this problem.

Our method relies on extensive computation in the subprogram for constructing the inequality \( {\Qo}(C, D) > 0 \), so it is efficient mainly when the input polynomials are  lower degrees. 
We have also applied our method to solve the SOS problem for multivariate polynomials \cite{Huang2020,HZYR2024}, but it is still unable to handle more complicated cases, 
such as
$$
\text{delzell}(x_1, x_2, x_3, x_4) = x_1^4 x_2^2 x_4^2 + x_2^4 x_3^2 x_4^2 + x_1^2 x_3^4 x_4^2 - 3x_1^2 x_2^2 x_3^2 x_4^2 + x_3^8.
$$
Nevertheless, our method is at least helpful for situations where a rigorous proof of positive definiteness is required. 
\smallskip

Finally, the authors suggest referring to our method as the "L-SOS Algorithm" in any future citations, to reflect that in this algorithm, the indeterminates of the lower-triangular matrix are divided into three groups. The indeterminates in the first column and last row, which form the L-shape, are treated as the main variables and represented as rational functions of the other variables (i.e., diagonal variables and core variables).

\newpage 

%
%
%
%

\appendix
\newpage

\end{document}